\documentclass[pra,
twocolumn,
superscriptaddress,nofootinbib,longbibliography]{revtex4-2}

\usepackage[pdftex]{graphicx}
\usepackage{amsmath,amsfonts,amsbsy,amssymb,amsthm,mathtools}
\usepackage{mathrsfs}
\usepackage{epstopdf}
\usepackage{ulem,color}
\usepackage{enumitem}
\usepackage{siunitx}
\usepackage{amssymb,amsmath}
\usepackage{ulem,color}
\usepackage{MnSymbol}
\usepackage{environ}
\usepackage{xcolor}

\usepackage{tikz}
\usetikzlibrary{quantikz}
\usetikzlibrary{automata,arrows,positioning,calc}

\usepackage{pgfplots}
\usepackage{pstool}
\usepackage{tikzsymbols}
\usetikzlibrary{plotmarks,spy}

\newtheorem{corollary}{Corollary}
\newtheorem{theorem}{Theorem}

\theoremstyle{definition}
\newtheorem{definition}{Definition}

\usepackage{hyperref}
\usepackage[T1]{fontenc}
\usepackage{mwe}
\usepackage{tikz}
\usetikzlibrary{quantikz}
\newcommand{\ketbra}[1]{|{#1}\>\mkern-4mu\<{#1}|}
\newcommand{\tr}{\textup{Tr}}
\renewcommand{\>}{\rangle}
\newcommand{\<}{\langle}
\newcommand{\C}{{\mathbb{C}}} %
\newcommand{\R}{{\mathbb{R}}} %

\renewcommand{\ll}{\llangle}
\newcommand{\rr}{\rrangle}
\newcommand{\spn}{\textrm{Span}}

\newcommand{\CQT}{Centre for Quantum Technologies, National University of Singapore, 3 Science Drive 2, Singapore 117543.\looseness=-1}
\newcommand{\NTU}{Nanyang Quantum Hub, School of Physical and Mathematical Sciences, Nanyang Technological University, Singapore 639673.\looseness=-1}
\newcommand{\IHPC}{A*STAR Quantum Innovation Centre (Q.InC), Institute of High Performance Computing (IHPC), Agency for Science, Technology and Research (A*STAR), 1 Fusionopolis Way, \#16-16 Connexis, Singapore, 138632, Republic of Singapore.\looseness=-1}

\newcommand{\CQuERE}{Centre for Quantum Engineering, Research and Education, TCG CREST, Sector V, Salt Lake, Kolkata 700091, India.\looseness=-1}

\begin{document}

\normalem
\newlength\figHeight 
\newlength\figWidth

\title{Strategic Code: A Unified Spatio-Temporal Framework for Quantum Error-Correction}

\author{Andrew Tanggara}
\email{andrew.tanggara@gmail.com}
\affiliation{\CQT}
\affiliation{\NTU}

\author{Mile Gu}
\email{mgu@quantumcomplexity.org}
\affiliation{\NTU}
\affiliation{\CQT}

\author{Kishor Bharti}
\email{kishor.bharti1@gmail.com}
\affiliation{\IHPC}
\affiliation{\CQuERE}

\date{\today}

\begin{abstract}
Quantum error-correcting code (QECC) is the central ingredient in fault-tolerant quantum information processing.
An emerging paradigm of dynamical QECC shows that one can robustly encode logical quantum information both temporally and spatially in a more resource-efficient manner than traditional QECCs.
Nevertheless, an overarching theory of how dynamical QECCs achieve fault-tolerance is lacking.
In this work, we bridge this gap by proposing a unified spatio-temporal QECC framework called the ``strategic code'' built around an ``interrogator'' device which sequentially measures and evolves the spatial QECC in an adaptive manner based on the ``quantum combs'' formalism, a generalization of the channel-state duality. 
The strategic code covers all existing dynamical and static QECC, as well as all physically plausible QECCs to be discovered in the future, including those that involve adaptivity in its operational dynamics.
Within this framework, we show an algebraic and an information-theoretic necessary and sufficient error-correction conditions for a strategic code, which consider spatially and temporally correlated errors.
These conditions include the analogous known static QECC conditions as a special case.
Lastly, we also propose an optimization-theoretic approach to obtain an approximate strategic code adapting to a correlated error.

\end{abstract}

\maketitle

The susceptibility of quantum systems to noise has been a major obstacle in achieving the advantages offered by quantum information processing tasks over their classical counterparts.
A quantum error-correcting code (QECC) overcomes this problem by redundantly encoding quantum information in a noise-robust manner.
However conventional QECCs involve operations on many-body quantum system to encode and decode logical information, which are notoriously resource-intensive.
The novel paradigm of dynamical QECC offers a promising solution by utilizing the temporal dimension to encode and decode logical information, thus easing this demanding requirement.
Many dynamical QECC has been proposed, with the most popular being the Floquet codes~\cite{hastings2021dynamically,gidney2021fault,vuillot2021planar,haah2022boundaries,gidney2022benchmarking,davydova2023floquet,wootton2022measurements,higgott2023constructions,zhang2023x,aasen2023measurement,paetznick2023performance,fahimniya2023hyperbolic,dua2024engineering} which measurement sequence is performed periodically, although other non-periodic codes have also been explored~\cite{bacon2017sparse,gottesman2022opportunities,delfosse2023spacetime,mcewen2023relaxing,berthusen2023partial,davydova2023quantum,fu2024error,bombin2023unifying,kesselring2024anyon}.

Despite the remarkable progress in dynamical QECC research, an overarching theory of error-correction for QECCs considering both spatial and temporal encoding that is analogous to its static counterpart~\cite{knill1997theory,schumacher1996quantum,nielsen1998information,cerf1997information,nielsen2007algebraic,gottesman1997stabilizer,kribs2005unified}, has been largely unexplored.
In order to analyze the error-correction capability of a QECC with respect to resources that it uses, both spatially and temporally, such theory is imperative.
In this work we bridge this gap by proposing a QECC framework called the ``strategic code'' which unifies all existing dynamical and static QECCs, as well as all physically plausible QECCs to be discovered.
The novelty of the strategic code lies in the ``interrogator'' device which captures any set of operations performed both spatially and temporally between the encoding and decoding stage, completing the conceptual gap operationally between static and dynamical QECC paradigms.
The strategic code framework also generalize existing QECC paradigms by accommodating temporal operational adaptivity of the code and  the effect of the most general class of noise with both spatial and temporal (non-Markovian) correlations~\cite{white2020demonstration,milz2021quantum,rivas2014quantum,li2018concepts,sakuldee2018non,aharonov2006fault,nickerson2019analysing,ng2009fault,preskill2012sufficient}.

Within the strategic code framework for QECCs with an interrogator that maintains a classical memory and for general error models (which may exhibit correlations), we show necessary and sufficient error-correction conditions, as well as formulate a multi-convex optimization problem to obtain an approximate code.
Our conditions are presented in two equivalent forms: algebraically (Theorem~\ref{thm:algebraic_KL_condition}) and an information-theoretically (Theorem~\ref{thm:info_theoretic_condition}).
Due the generality of our framework, algebraic and information-theoretic necessary and sufficient conditions for static QECC~\cite{knill1997theory,nielsen1998information,schumacher1996quantum,cerf1997information} are included as special cases. 
Since strategic code subsumes notable QECC frameworks, such as the sequential Pauli measurements framework~\cite{fu2024error}, ZX calculus framework~\cite{bombin2023unifying}, and anyon condensation framework~\cite{kesselring2024anyon}, these conditions serve as a guide in a code construction within these frameworks.

\section{General QECC Scenario}\label{sec:general_QECC}

\begin{figure*}
    \centering
    \
    \includegraphics{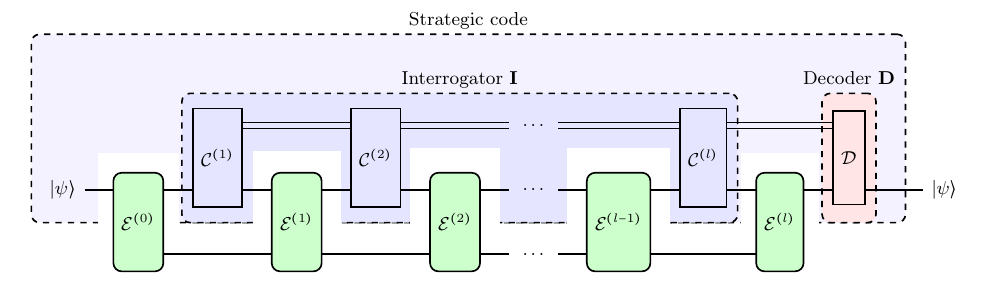}
    \caption{
    General QECC Scenario.
    At round $0$, error $\mathcal{E}^{(0)}$ is inflicted to code state $|\psi\>$ that belongs to the initial code space $\mathscr{S}_{Q_0}$, evolving it to error-inflicted subspace $\mathscr{S}_{Q_0'}$.
    Then an interrogator device $\mathbf{I}$ is applied to the code, performing quantum operations on the code $l$ times in sequence.
    At each round $r\in\{1,\dots,l\}$, evolution on the code is performed by a quantum operation $\mathcal{C}^{(r)}$, giving a new code space $\mathscr{S}_{Q_r}$ followed by error map $\mathcal{E}^{(r)}$ further evolving the code space $\mathscr{S}_{Q_r}$ to error-inflicted subspace $\mathscr{S}_{Q_r'}$.
    The interrogator maintains a classical memory (illustrated by double-wires) which keeps information about previous events (e.g. measurement outcome) and determines the choice of quantum operation performed at each round.
    At the start of round $r$, if the interrogator's  memory state is $m_{r-1}$ then it performs quantum operation $\mathcal{C}_{m_{r-1}}^{(r)}$.
    After performing the operation, it the updates the state of the classical memory register $m_{r-1}\mapsto m_r$ where $m_r$ is the updated memory state that determines the subsequent operations.
    After the final round $l$, decoding $\mathbf{D}$ is performed based on the final memory state $m_l$, where channel $\mathcal{D}_{m_l}$ is applied to the code system to restore the initial code state $|\psi\>$.}
    \label{fig:dynamical_QECC}
\end{figure*}

Error-correcting code scenario can always be described as interactions between a code and some noise changing the state of a physical system where logical information is encoded in.
Conventionally, the essential two ingredients of a code are: (1) an operation encoding the logical information into a physical system and (2) an operation that recovers logical information from any error caused by noise that occurs in between.
In a quantum error-correcting code (QECC) scenario, further interplay between the code and noise can take place between encoding and decoding, allowing more interesting implications on the error-correction process due to inherent quantum-mechanical effects.

In general, a QECC scenario consists of three stages:
The first stage being the encoding stage, the final stage being the decoding stage, and between them, any set of operations performed by the code, interacting with the noise (see Fig.~\ref{fig:dynamical_QECC}).
We simplify this by starting with an initial codespace $\mathscr{S}_{Q_0}$ instead of an encoding map (as an encoding map uniquely defines the initial codespace).
Then a set of operations $\mathbf{I}$, called the \textit{interrogator}, is performed by the code in between encoding and decoding.
The interrogator performs $l\geq0$ rounds of \textit{check operations} where the check operation $\mathcal{C}^{(r)}$ at round $r$ may be chosen (from a set of allowed operations) \textit{adaptively} based on events happened in the previous rounds stored in a classical memory register.
When the classical memory register is in a \textit{memory state} $m_{r-1}$, then a check operation $\mathcal{C}_{m_{r-1}}^{(r)}$ is performed.
Between encoding and decoding, we also have errors $\mathbf{E}$ being inflicted on the code right after the encoding stage and after each round of operation.
Errors can generally have spatial correlations within each round and temporal (non-Markovian) correlations across rounds.
Lastly in the decoding stage, the decoding procedure $\mathbf{D}$ consists of multiple decoders that the coder can choose from based on the information stored in the classical memory $m_l$ after the last round of operation.
A successful QECC procedure recovers logical information encoded in the initial codespace as codestate $|\psi\>$.

Although here we focus on an interrogator where only classical memory storage is allowed, one can easily generalize this to an interrogator where quantum memory is involved.
Such interrogator could model a scenario where a small-size quantum system can be used reliably to store some information about the code across time alongside classical memory to be used in the decoding stage.
Particularly, this generalization reduces to the entanglement-assisted QECC (EAQECC)~\cite{brun2006correcting,hsieh2007general,brun2014catalytic} when we set the number of rounds $l=0$, see Appendix~\ref{app:interrogator_quantum_memory} for details.

More formally, the entirety of an $l$-round QECC scenario is defined by the following objects:
\begin{enumerate}
    \item A sequence of codespaces $\mathscr{S}_{Q_0},\mathscr{S}_{Q_0'},\dots,\mathscr{S}_{Q_l},\mathscr{S}_{Q_l'}$ which are a subspace of $d$ dimensional complex vector space $\C^d$ and spaces of bounded linear operators from $\C^d$ to $\C^d$, $\mathscr{H}_{Q_0},\mathscr{H}_{Q_0'},\dots,\mathscr{H}_{Q_l},\mathscr{H}_{Q_l'},\mathscr{H}_D$. 
    The codespaces and operator spaces depends on operations performed in the two stages defined below, while the initial codespace $\mathscr{S}_{Q_0}$ is defined independent of them.

    \item An interrogator $\mathbf{I}$ consists of a sequence of \textit{check instruments} $\mathcal{C}^{(1)},\, \mathcal{C}^{(2)}:=\{\mathcal{C}_{m_1}^{(2)}\}_{m_1},\, \dots,\, \mathcal{C}^{(l)}:=\{\mathcal{C}_{m_{l-1}}^{(l)}\}_{m_{l-1}}$ 
    where $\mathcal{C}_{m_{r-1}}^{(r)} := \{\mathcal{C}_{m_r|m_{r-1}}^{(r)}\}_{m_r}$ and $\mathcal{C}^{(1)} := \{\mathcal{C}_{m_1}^{(1)}\}_{m_1}$ are a check instrument for round $r>1$ and round $r=1$, respectively.
    Each $\mathcal{C}_{m_r|m_{r-1}}^{(r)}:\mathscr{H}_{Q_{r-1}'}\rightarrow\mathscr{H}_{Q_r}$ is a completely-positive (CP) map such that $\sum_{m_r} \mathcal{C}_{m_r|m_{r-1}}^{(r)}$ is trace preserving (TP).
    Check instrument $\mathcal{C}_{m_{r-1}}^{(r)}$ may perform a deterministic operation on the code (e.g. a unitary) in which $\mathcal{C}_{m_{r-1}}^{(r)}$ consist of only one element, or a probabilistic operation where each of its element maps an initial codestate to a post-measurement codestate.

    \item Decoder $\mathbf{D}=\{\mathcal{D}_{m_l}\}_{m_l}$ where \textit{decoding channel} $\mathcal{D}_{m_l}:\mathscr{H}_{Q_l'} \rightarrow \mathscr{H}_D$ is a CPTP map that given classical information $m_l$ in the classical memory register recovers the initial code state $|\psi\>_{Q_0}$ (what this precisely means will be defined shortly in Definition~\ref{def:dynamical_QECC}).

    \item Error $\mathbf{E}$ consists of a sequence of \textit{error maps} $\mathcal{E}^{(0)},\mathcal{E}^{(1)},\dots,\mathcal{E}^{(l)}$ where $\mathcal{E}^{(r)}:\mathscr{H}_{Q_r}\otimes\mathscr{H}_{E_{r-1}}\rightarrow\mathscr{H}_{Q_r'}\otimes\mathscr{H}_{E_r}$ is a CP trace non-increasing map defined by $\mathcal{E}^{(r)}(\rho)=\sum_{e_r}E_{e_r} \rho E_{e_r}^\dag$ where bounded linear operator $E_{e_r}:\mathscr{S}_{Q_r}\otimes\mathscr{S}_{E_{r-1}}\rightarrow\mathscr{S}_{Q_r'}\otimes\mathscr{S}_{E_r}$ being a Kraus operator for $\mathcal{E}^{(r)}$ for $r\geq1$ and $E^{(0)}:\mathscr{S}_{Q_0}\rightarrow\mathscr{S}_{Q_0'}\otimes\mathscr{S}_{E_0}$.
    Spaces labeled by $E_{r-1}$ and $E_r$ are the systems storing any temporal correlations in the noise environment from round $r-1$ to round $r$ and from round $r$ to round $r+1$, respectively.
    In the case of uncorrelated error sequence we simply have $\mathcal{E}^{(r)}:\mathscr{H}_{Q_r}\rightarrow\mathscr{H}_{Q_r'}$, with Kraus operators of the form $E_{e_r}:\mathscr{S}_{Q_r} \rightarrow\mathscr{S}_{Q_r'}$.
\end{enumerate}

We remark that in the description above it is assumed that the dimension $d$ of the quantum system where the code lives is always the same at all rounds.
However, this assumption is only for convenience and is not assumed in our proofs thus can be relaxed to round-dependent dimensions.
Namely, code spaces $\mathscr{S}_{Q_r},\mathscr{S}_{Q_r'}$ being a subspace of $\C^{d_r}$ and $\mathscr{H}_{Q_r},\mathscr{H}_{Q_r'}$ being spaces of bounded linear operators from $\C^{d_r}$ to $\C^{d_r}$.

Completely positive map $\mathcal{C}_{m_r|m_{r-1}}^{(r)}$ corresponds to mapping multiple measurement outcomes arising from measurement setting defined by memory state $m_{r-1}$. 
This can be formally described as POVM $\{M_{o_r|m_{r-1}}\}_{o_r}$ defined by the Kraus operators $C_{o_r|m_{r-1}}^{(r)}$ of CPTP map $\mathcal{C}_{m_{r-1}}^{(r)} = \sum_{m_r} \mathcal{C}_{m_r|m_{r-1}}^{(r)}$ defined by the instrument $\{\mathcal{C}_{m_r|m_{r-1}}^{(r)}\}_{m_r}$.
Namely, $M_{o_r|m_{r-1}} = C_{o_r|m_{r-1}}^{(r)\dag} C_{o_r|m_{r-1}}^{(r)}$ and $\mathcal{C}_{m_r|m_{r-1}}^{(r)}(\rho) = \sum_{o_r} C_{o_r|m_{r-1}}^{(r)} \rho C_{o_r|m_{r-1}}^{(r)\dag}$.
The memory state $m_r$ is defined by some \textit{memory update function} $f_r$ that maps the measurement outcomes $o_r$ and memory state $m_{r-1}$ to memory state $m_r$, namely $\mathcal{C}_{m_{r-1}}^{(r)} = \sum_{m_r} \sum_{o_r:f_r(o_r,m_{r-1})=m_r} \mathcal{C}_{o_r|m_{r-1}}^{(r)}$ where $\mathcal{C}_{o_r|m_{r-1}}^{(r)}(\rho) = C_{o_r|m_{r-1}}^{(r)} \rho C_{o_r|m_{r-1}}^{(r)\dag}$.
Thus, $\mathcal{C}_{m_r|m_{r-1}} = \sum_{o_r:f_r(o_r,m_{r-1})=m_r} \mathcal{C}_{o_r|m_{r-1}}^{(r)}$.
Note that here we consider time dependent memory update functions $f_1,\dots,f_l$, for full generality, but of course one can instead consider a time-independent function $f$ used in every round. 
Both measurement outcome $o_r$ and measurement setting $m_{r-1}$ should give spatial and temporal information about error occurrence.
The memory state $m_{r-1}$ at the start of round $r\in\{2,\dots,l\}$ determines the choice of instrument $\mathcal{C}_{m_{r-1}}^{(r)}$ applied in that round, whereas memory state $m_l$ at final round $l$ is used to choose which decoder $\{\mathcal{D}_{m_l}\}_{m_l}$ is used to recover the initial codestate.
Each final memory state $m_l$ corresponds to a unique set $O_{m_l}$ (defined by memory update functions $f_1,\dots,f_l$) containing all check measurement outcome sequence $o=o_1,o_2,\dots,o_l$ where there exists a sequence of memory states $m_1,\dots,m_{l-1}$ such that $m_l=f_l(o_l,m_{l-1})$ and $m_{l-1}=f_{l-1}(o_{l-1},m_{l-2})$, ..., $m_1=f_1(o_1)$.
Hence we can define a function $f^*$ as $f_l^*(o) = f_l(o_l,f_{l-1}(o_{l-1},\dots f(o_1)\dots))$ which maps an outcome sequence $o$ to a final memory state.
So we can express in symbols, $O_{m_l} = \{o : f_l^*(o)=m_l\}$. 

Note that codespace $\mathscr{S}_{Q_r}$ at round $r$ is determined by the choice of check instrument $\mathcal{C}_{m_{r-1}}^{(r)}$ and the outcome $o_r$ from the measurement defined by it.
Namely, $\mathcal{C}_{o_r|m_{r-1}}^{(r)}$ maps the state of the codespace $\mathscr{S}_{Q_{r-1}'}$ (the codespace after error map $\mathcal{E}^{(r-1)}$ at the previous round) to codespace $\mathscr{S}_{Q_r}$.
Also, both codespaces $\mathscr{S}_{Q_r}$ and $\mathscr{S}_{Q_{r-1}'}$ may depend on the check instruments and error maps in the previous rounds.\footnote{For Floquet codes this corresponds to the instantaneous stabilizer groups, which is a stabilizer group at a particular round $r$ defined by the check measurement outcome at that round and as well as the outcomes and errors occurring up to that round.}

Interrogator $\mathbf{I}$, error $\mathbf{E}$, and decoder $\mathbf{D}$ in a QECC scenario can be represented by the quantum combs formalism~\cite{chiribella2009theoretical,chiribella2008quantum,milz2021quantum,giarmatzi2021witnessing,pollock2018non,gutoski2007toward,oreshkov2012quantum}\footnote{Also known as process tensor in~\cite{pollock2018non} or quantum strategy in~\cite{gutoski2007toward}, and more generally, process matrix in~\cite{oreshkov2012quantum} where exotic causal ordering can be exhibited.}.
It has the advantage of compactly representing temporally-correlated sequence of dynamics on a quantum system as a positive semidefinite operator by generalizing the Choi-Jamio{\l}kowski isomorphism~\cite{choi1975completely,jamiolkowski1972linear,jiang2013channel} which holds the complete information about the temporal dynamics.
Quantum combs has found many applications such as quantum causal inference~\cite{bai2020efficient,costa2016quantum}, metrology~\cite{chiribella2012optimal,liu2023optimal}, interactive quantum games~\cite{gutoski2007toward}, open quantum systems~\cite{luchnikov2019simulation}, quantum cryptography~\cite{gutoski2018fidelity}, and quantum communication~\cite{kristjansson2020resource}.
However, to our knowledge no application of quantum combs has been made in the context of QECC before.
In the following we describe the quantum combs representation for the QECC scenario.
More technical details of the quantum combs representation can be found in Appendix~\ref{app:choi_jamiolkowski_isomorphism_link_product}.

\subsection{Strategic code and quantum combs representation}

In an $l$-round QECC scenario, the entirety of how logical information is preserved can be described by the initial codespace $\mathscr{S}_{Q_0}$ and the sequence of check instruments $\mathcal{C}^{(1)},\mathcal{C}^{(2)},\dots,\mathcal{C}^{(l)}$.
These two objects made up the strategic code, defined formally in the following.

\begin{definition}
    An $l$-round \textit{strategic code} is defined by a tuple $(\mathscr{S}_{Q_0},\mathbf{I})$ where $\mathscr{S}_{Q_0}$ is the \textit{initial codespace} which is a subspace of $\C^d$, and $\mathbf{I} = \{\mathbf{I}_{m_l}\}_{m_l}$ is a collection of positive semidefinite operators in $\mathscr{H}_{Q_0'} \otimes ( \bigotimes_{r=1}^l \mathscr{H}_{Q_{r-1}'}\otimes\mathscr{H}_{Q_r} )$ called the \textit{interrogator}.
    Interrogator $\mathbf{I}$ describes all possible sequences of check instruments $\{\mathcal{C}^{(1)},\mathcal{C}_{m_1}^{(2)},\dots,\mathcal{C}_{m_{l-1}}^{(l)}\}_{m_1,\dots,m_{l-1}}$ along their with temporal dependence defined by functions $f_1,\dots,f_l$.
    An element $\mathbf{I}_{m_l}$ of an interrogator is called an \textit{interrogator operator}, which is a quantum comb describing the sequences of CP maps $\mathcal{C}_{m_1}^{(1)},\mathcal{C}_{m_2|m_1}^{(2)},\dots,\mathcal{C}_{m_l|m_{l-1}}^{(l)}$ that ends in final memory state $m_l$.
    An interrogator operator $\mathbf{I}_{m_l}$ takes the form of 
    \begin{equation}\label{eqn:eigenvector_interrogator}
        \mathbf{I}_{m_l} = \sum_{o\in O_{m_l}} |C_{m_l,o}\rr\ll C_{m_l,o}|
    \end{equation}
    where $O_{m_l}$ is the set of check measurement outcome sequences $o=o_1,\dots,o_l$ resulting in final memory state $m_l$.
    Here $|C_{m_l,o}\rr = |C_{o_l|m_{l-1}}^{(l)}\rr\otimes\dots\otimes |C_{o_1}^{(1)}\rr$ (with $m_r = f_r(o_r,m_{r-1})$) is an eigenvector of interrogator operator $\mathbf{I}_{m_l}$, and $|C_{o_r|m_{r-1}}^{(r)}\rr$ is the vectorized representation of the Kraus operator $C_{o_r|m_{r-1}}^{(r)}$.
\end{definition}

Note that $\mathscr{H}_{Q_0'} \otimes ( \bigotimes_{r=1}^l \mathscr{H}_{Q_{r-1}'}\otimes\mathscr{H}_{Q_r} )$ is the tensor product of bounded linear operator space of the inputs and outputs of the sequence of CP maps of the check instruments.
An interrogator describes all possible ``trajectories'' of how the code evolves according to the sequence of check measurement outcomes and operations performed based on previously obtained outcomes.
As an example, if the initial codestate is $\rho^{Q_0}$ (round $r=0$) and in round $r=1$ a measurement is performed, an outcome $0$ will result in post-measurement codestate $\rho_0^{Q_1}$ while an outcome $m_1\neq0$ codestate $\rho_{m_1}^{Q_1}$ which is generally is not equal to $\rho_0^{Q_1}$.
Outcome $m_1$ also determines which operation is performed in the next round ($r=2$).
So if a measurement at $r=2$ depending on outcome $m_1$ gives an outcome $m_2$, then we obtain codestate $\rho_{m_2|m_1}^{Q_2}$, where the label $m_2|m_1$ describes the ``trajectory''.
If decoding is performed on $\rho_{m_2|m_1}^{Q_2}$, then the decoder will use $m_2$ to recover the initial state.
For a diagram illustrating the trajectories of the code induced by the operations performed by the interrogator, see Fig.~\ref{fig:interrogator}.

\begin{figure}
\centering
\includegraphics[width=1\columnwidth]{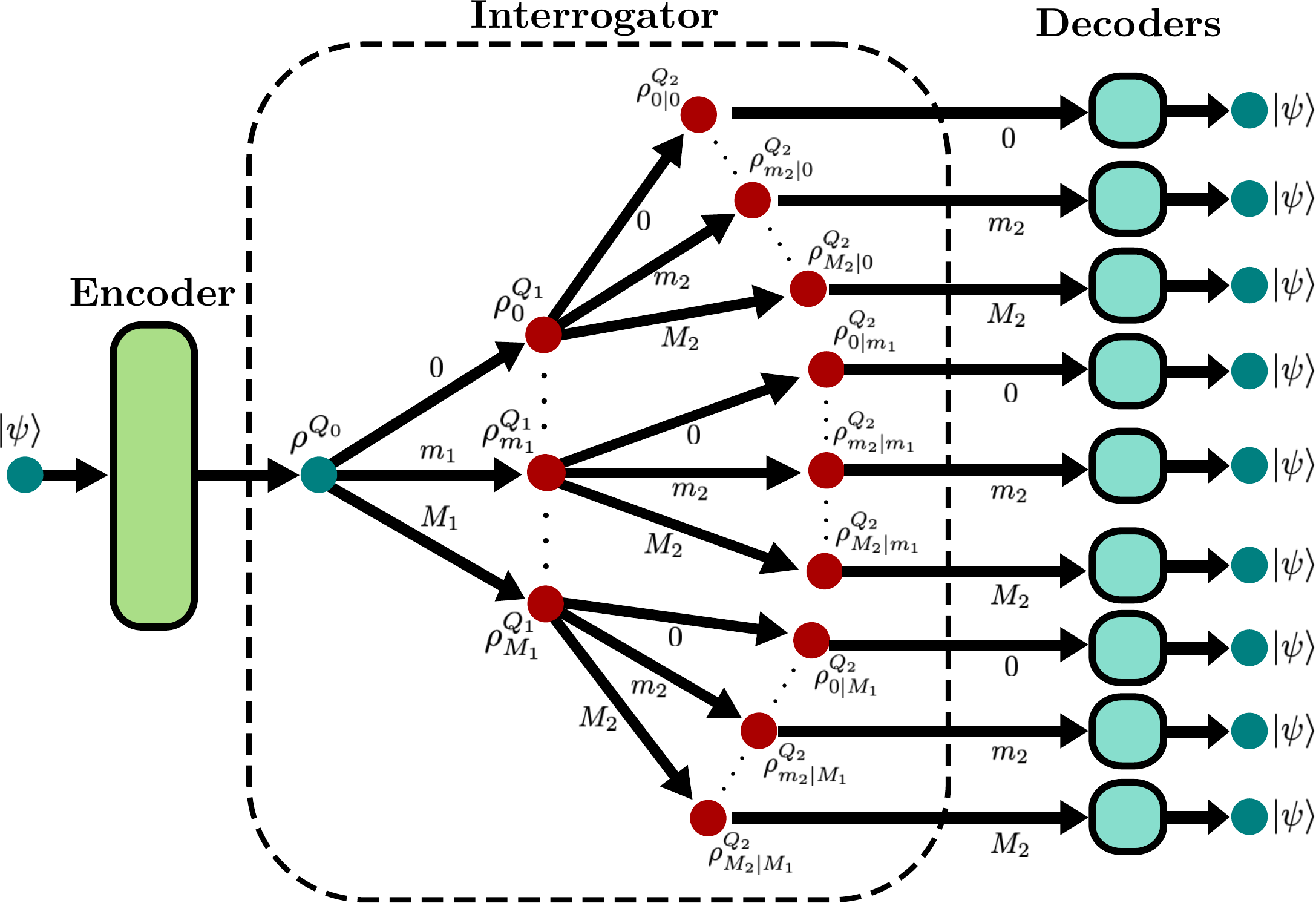}
\caption{Trajectories of the code evolution in an interrogator.}
\label{fig:interrogator}
\end{figure}

Now we turn to how decoding and errors are described in an $l$-round QECC scenario.
Decoding channel $\mathcal{D}_{m_l}$ is described by each of its Choi-Jamiolkowski representation $\mathbf{D}_{m_l}$, which is a positive semidefinite operator in $\mathscr{H}_{Q_l'}\otimes\mathscr{H}_D$.
Similarly, we can also represent the sequence of error CP maps $\mathcal{E}^{(0)},\dots,\mathcal{E}^{(l)}$ as positive semidefinite operator
\begin{equation}
    \mathbf{E} = \sum_e |E_e\rr\ll E_e|
\end{equation}
where  $e=e_0,e_1,\dots,e_l$ indicates an error sequence.
Here $|E_e\rr$ is the vectorized representation of the sequence of Kraus operators $E_{e_0},\dots,E_{e_l}$ of the error maps.

The entire interaction between an interrogator operator ending in final memory state $m_l$ and the error maps can be described as $\mathbf{E} \ast \mathbf{I}_{m_l}$ where "$\ast$" is an associative dyadic operation between two linear operators $A\in\mathscr{H}_A\otimes\mathscr{H}_C$ and $B\in\mathscr{H}_C\otimes\mathscr{H}_B$ known as the \textit{link product}~\cite{chiribella2009theoretical}, which can be thought of as a generalization of the Hilbert-Schmidt inner product.
The link product is defined as $A\ast B := \tr_C(( A^{\top_C}\otimes I_B)(I_A\otimes B))$ where $\tr_C$ is partial trace over operator space $\mathscr{H}_C$, $\cdot^{\top_C}$ is the partial transpose over $\mathscr{H}_C$, and $I_A,I_B$ are identity operators of $\mathscr{H}_A$ and $\mathscr{H}_B$, respectively.
Note that if the operator space $\mathscr{H}_C$ has trivial dimension then $A\ast B = A\otimes B$, and if $\mathscr{H}_A,\mathscr{H}_B$ have a trivial dimension then $A\ast B = \tr(A^\top B)$.

Some of our results are expressed in term of the explicit error sequence $e$ described by $|E_e\rr\ll E_e|$ instead of $\mathbf{E}$.
Hence at times we will describe the sequence of error maps as the collection of the rank-1 operators describing the error sequences, $\mathfrak{E}:=\{|E_e\rr\ll E_e|\}_e$.
In both the link product form and the rank-1 vector form, we can express the complete interaction between the an interrogator operator and an error sequence $e$ as
\begin{equation}\label{eqn:error_check_kraus_repr}
\begin{aligned}
    \mathbf{E}\ast\mathbf{I}_{m_l} &= \sum_e |E_e\rr\ll E_e| \ast \mathbf{I}_{m_l} \\
    &= \sum_e \sum_{o\in O_{m_l}} |K_{e,m_l,o}\rr\ll K_{e,m_l,o}| \;,
\end{aligned}
\end{equation}
where $|K_{e,m_l,o}\rr$ is the vectorized form of Kraus operator $K_{e,m_l,o}$ defined by the Kraus operators of the check instruments and error maps $E_{e_l},C_{o_l|m_{l-1}}^{(l)},\dots,C_{o_1}^{(1)},E_{e_0}$ (where the check instrument $\mathcal{C}_{o_r|m_{r-1}}^{(r)}$ in round $r$ implicitly performs an identity map on the environment $\mathscr{H}_{E_{r-1}}$, for explicit definition see Appendix~\ref{app:choi_jamiolkowski_isomorphism_link_product}).

\subsection{Interrogator in existing codes}

A sequence of quantum operations that temporally evolve the spatial encoding of logical information while also extracting error syndromes is the central ingredient in dynamical QECCs such as spacetime codes~\cite{bacon2017sparse,gottesman2022opportunities,delfosse2023spacetime}, Floquet codes~\cite{hastings2021dynamically,gidney2021fault,vuillot2021planar,haah2022boundaries,gidney2022benchmarking,davydova2023floquet,wootton2022measurements,higgott2023constructions,zhang2023x,aasen2023measurement,paetznick2023performance,fahimniya2023hyperbolic,dua2024engineering}, and dynamical code~\cite{fu2024error}. 
The notion of interrogator captures 
This sequence of operations in these codes can be formulated as an interrogator.

In an $n$-qubit spacetime code, generally a round of operation performed by the interrogator consists of Clifford gates and Pauli measurements on disjoint subsets of the $n$-qubits.
Since the operations are not adapting to events in preceding rounds (measurements, etc.), here the interrogator simply stores all measurement outcomes $o$ to be used at the decoding round.
For more detailed discussions on interrogator formulation of spacetime code see Appendix~\ref{app:spacetime_code}.

On the other hand in a Floquet code and a dynamical code, generally a round of operation by the interrogator consists of commuting Pauli measurements.
Similar to the spacetime code, the sequence of operations performed on the code is predetermined and the interrogator simply stores all measurement outcomes $o$ to be used at the decoding round.
In Floquet code, however, Pauli measurements are performed in cycles.
For example in the Hastings-Haah honeycomb code~\cite{hastings2021dynamically,gidney2021fault}, one performs a 3-cycle of Pauli $X$, followed by Pauli $Y$, then Pauli $Z$, which in principle can be performed cyclically indefinitely.
However in the finite time window where measurement syndromes can be revealed from the measurements, error-correction analysis using an interrogator consisting of measurements in this time window can be useful.
For more detailed discussion on the Hastings-Haah honeycomb code in the interrogator form, see Appendix~\ref{app:hastings_haah_interrogator}.

\section{General Error-Correction Conditions}

We now formalize the notion of correctability of an error $\mathbf{E}$ by a strategic code $(\mathscr{S}_{Q_0},\mathbf{I}=\{\mathbf{I}_{m_l}\}_{m_l})$.

\begin{definition}\label{def:dynamical_QECC}
    Strategic code $(\mathscr{S}_{Q_0},\mathbf{I})$ corrects error $\mathbf{E}$ if there exists decoding channels $\{\mathcal{D}_{m_l}\}_{m_l}$ such that for all $|\psi\>\in\mathscr{S}_{Q_0}$,
    \begin{equation}\label{eqn:dynamical_QECC}
    \begin{aligned}
        \mathcal{D}_{m_l}(\mathbf{E} \ast \mathbf{I}_{m_l} \ast \ketbra{\psi}) = \lambda_{m_l} \ketbra{\psi}
    \end{aligned}
    \end{equation}
    for some constant $\lambda_{m_l}\in\R$.
\end{definition}

We note that how recovery of the initial codestate $|\psi\>$ is being defined to be up to a constant independent of $|\psi\>$ is an artifact from the fact that the error $\mathbf{E}$ and the interrogator operator $\mathbf{I}_{m_l}$ are not necessarily trace preserving.
Namely when there is only one reachable final memory state $m_l$ with probability one given error $\mathbf{E}$ and initial codestate $|\psi\>$ then the constant is independent of $m_l$.
Additionally when all error maps $\mathcal{E}^{(0)},\dots,\mathcal{E}^{(l)}$ are trace preserving then we have $\lambda_{m_l}=1$.
Now given a precise definition of a successful recovery for error-correction, in the following we give two equivalent necessary and sufficient conditions in an algebraic form and in an information-theoretic form.

\subsection{Algebraic error-correction condition}

Now we state the algebraic necessary and sufficient error-correction condition for a strategic code.
For notational simplicity, below we suppress the identity operators for operator multiplications to match the dimensions, e.g. for operators $M\in\mathscr{H}_A\otimes\mathscr{H}_C$ and $N\in\mathscr{H}_C\otimes\mathscr{H}_B$ we write $MN$ when we mean $(M\otimes I_B)(I_A\otimes N)$.

\begin{theorem}\label{thm:algebraic_KL_condition}
    A strategic code $(\mathscr{S}_{Q_0},\mathbf{I})$ corrects $\mathfrak{E}$ if and only if
    \begin{equation}\label{eqn:choi_oper_dynamical_KL_condition}
    \begin{gathered}
        \ll E_{e'}| \big( |C_{m_l}\rr\ll C_{m_l,o}| \otimes |j\>\<i| \big) |E_e\rr = \lambda_{e',e,m_l,o} \delta_{j,i}
    \end{gathered}
    \end{equation}
    for a constant $\lambda_{e',e,m_l,o}\in\C$, and for all $m_l$, all check measurement outcome sequence $o\in O_{m_l}$, all pairs of error sequences $e,e'$, and all $i,j$.
    
    Here $|C_{m_l}\rr = \sum_{o\in O_{m_l}} |C_{m_l,o}\rr$ and $|C_{m_l,o}\rr$ is an eigenvector of interrogator operator $\mathbf{I}_{m_l}$ as defined in eqn.~\eqref{eqn:eigenvector_interrogator}.
    Vectors $|i\>,|j\>$ are orthonormal basis vectors of initial codespace $\mathscr{S}_{Q_0}$.
\end{theorem}
\begin{proof}
    First assume that strategic code $(\mathscr{S}_{Q_0},\mathbf{I})$ corrects $\mathfrak{E}$, hence there is a set of decoding channels $\{\mathcal{D}_{m_l}\}_{m_l}$ such that~\eqref{eqn:dynamical_QECC} is satisfied.
    First note that we can express eqn.~\eqref{eqn:dynamical_QECC} in Kraus representation and use $\Pi_{Q_0}|\psi\> = |\psi\>$ where $\Pi_{Q_0}$ is a projector to initial codespace $\mathscr{S}_{Q_0}$ to obtain
    \begin{equation}
    \begin{aligned}
        &\mathcal{D}_{m_l}(\mathbf{E} \ast \mathbf{I}_{m_l} \ast \ketbra{\psi}) \\
        &= \sum_{e,o} \mathcal{D}_{m_l}(K_{e,m_l,o}\Pi_{Q_0}\ketbra{\psi}\Pi_{Q_0}K_{e,m_l,o}^\dag) \\
        &= \sum_{e,o,k} D_{k|m_l}K_{e,m_l,o}\Pi_{Q_0}\ketbra{\psi}\Pi_{Q_0}K_{e,m_l,o}^\dag D_{k|m_l}^\dag \\
        &= \lambda_{m_l} \Pi_{Q_0}\ketbra{\psi} \Pi_{Q_0} \;,
    \end{aligned}
    \end{equation}
    where $\{D_{k|m_l}\}_k$ are the Kraus operators of decoding channel $\mathcal{D}_{m_l}$.
    By the non-uninqueness of Kraus representation (see e.g.~\cite{nielsen2010quantum}), both $D_{k|m_l}K_{e,m_l,o}$ and $\sqrt{\lambda_{m_l}}\Pi_{Q_0}$ can be thought of as Kraus representations of the composition of maps $\mathcal{D}_{m_l}\circ (\mathbf{E}\ast\mathbf{I}_{m_l}\ast(\cdot))$, where the latter Kraus representation consists of only one operator.
    Thus there must exist some complex number $\gamma_{k,e,o}^{(m_l)}$ such that $D_{k|m_l}K_{e,m_l,o}\Pi_{Q_0} = \gamma_{k,e,o}^{(m_l)}\Pi_{Q_0}$ for all $e,o,k$.
    Thus we have
    \begin{equation}
    \begin{aligned}
        &\ll E_{e'}| \big( |C_{m_l}\rr\ll C_{m_l,o}| \otimes |j\>\<i| \big) |E_e\rr \\
        &= \sum_{o',k} (\ll K_{e',m_l,o'}|j\>)  D_{k|m_l}^\dag D_{k|m_l} (\<i|K_{e,m_l,o}\rr) \\
        &= \sum_{o',k} \<j|\Pi_{Q_0} K_{e',m_l,o'}^\dag D_{k|m_l}^\dag D_{k|m_l} K_{e,m_l,o} \Pi_{Q_0}|i\> \\
        &= \sum_{o',k} \gamma_{k,e',o'}^{(m_l)*} \gamma_{k,e,o}^{(m_l)} \<j|i\> \\
        &= \Big(\sum_k \big(\sum_{o'}\gamma_{k,e',o'}^{(m_l)*}\big) \,\gamma_{k,e,o}^{(m_l)}\Big) \delta_{j,i} \\
        &= \lambda_{e',e,m_l,o} \delta_{j,i} \;,
    \end{aligned}
    \end{equation}
    where the second equality is obtained by $\sum_k D_{k|m_l}^\dag D_{k|m_l} = I_{Q_l'}$ and by $|K_{e,m_l,o}\rr = \ll C_{m_l,o}| E_{e}\rr$, the third equality by $\<i|K_{e,m_l,o}\rr = K_{e,m_l,o}|i\>$, and the fourth equality by the non-unique Kraus representation relation $D_{k|m_l}K_{e,m_l,o}\Pi_{Q_0} = \gamma_{k,e,o}^{(m_l)}\Pi_{Q_0}$.
    Thus we show the necessity of condition in Theorem~\ref{thm:algebraic_KL_condition}.

    Now assume that eqn.~\eqref{eqn:choi_oper_dynamical_KL_condition} holds and consider $\lambda_{e',e,m_l} = \sum_o \lambda_{e',e,m_l,o}$.
    Note that $\Lambda_{m_l} = [\lambda_{e',e,m_l}]_{e',e}$ is a Hermitian matrix since by eqn.~\eqref{eqn:choi_oper_dynamical_KL_condition}
    \begin{equation}
    \begin{aligned}
        &\lambda_{e',e,m_l}^* = \sum_o \lambda_{e',e,m_l,o}^* \\
        &= \sum_{o',o} (\ll E_{e'}|(|C_{m_l,o'}\rr\ll C_{m_l,o}| \otimes \ketbra{i})|E_e\rr)^* \\
        &= \sum_{o',o} \ll E_e|(|C_{m_l,o}\rr\ll C_{m_l,o'}| \otimes \ketbra{i})|E_{e'}\rr \\
        &= \sum_{o'} \lambda_{e',e,m_l,o'}^* \\
        &= \lambda_{e,e',m_l} \;.
    \end{aligned}
    \end{equation}
    Since matrix $\Lambda_{m_l}$ is Hermitian, it can be diagonalized to a diagonal matrix $[d_{e',e,m_l}]_{e',e}$ where $d_{e',e,m_l}=0$ if $e' \neq e$ as
    \begin{equation}\label{eqn:diag_constant_dynamical_KL_cond}
    \begin{aligned}
        d_{e',e,m_l} = \sum_{\Tilde{e},\Bar{e}} u_{e',\Tilde{e}}^* u_{\Bar{e},e} \lambda_{\Tilde{e},\Bar{e},m_l}
    \end{aligned}
    \end{equation}
    where $U= [u_{e',e}]_{e',e}$ is a unitary matrix.

    Now consider $|F_e\rr = \sum_{\Bar{e}} u_{\Bar{e},e} |E_{\Bar{e}}\rr$ so that
    \begin{equation}
    \begin{aligned}
        \sum_e \<g|F_e\rr\ll F_e|g'\> &= \sum_{e,\Tilde{e},\Bar{e}} u_{e,\Tilde{e}}^* u_{\Bar{e},e} \<g|E_{\Bar{e}}\rr\ll E_{\Tilde{e}}|g'\> \\
        &= \sum_{\Tilde{e},\Bar{e}} \delta_{\Tilde{e},\Bar{e}} \<g|E_{\Bar{e}}\rr\ll E_{\Tilde{e}}|g'\> \\
        &= \sum_e \<g|E_e\rr\ll E_e|g'\>
    \end{aligned}
    \end{equation}
    for any $|g\>,|g'\> \in (\bigotimes_{r=0}^{l-1} \mathscr{S}_{Q_r}\otimes\mathscr{S}_{Q_r'})\otimes\mathscr{S}_{Q_l}$.
    Therefore
    \begin{equation}
    \begin{aligned}
        \mathbf{E} \ast \mathbf{I}_{m_l} \ast |i\>\<j| &= \sum_e |E_e\rr\ll E_e| \ast \mathbf{I}_{m_l} \ast |i\>\<j| \\ 
        &= \sum_e |F_e\rr\ll F_e| \ast \mathbf{I}_{m_l} \ast |i\>\<j| \;,
    \end{aligned}
    \end{equation}
    and also
    \begin{equation}\label{eqn:transformed_error_opers_KL_cond}
    \begin{aligned}
        &\ll F_{e'}|\Big( |C_{m_l}\rr\ll C_{m_l,o}| \otimes |j\>\<i| \Big)|F_e\rr \\
        &= \sum_{\Tilde{e},\Bar{e}} u_{e',\Tilde{e}}^* u_{\Bar{e},e} \ll E_{\Tilde{e}}|\Big( |C_{m_l}\rr\ll C_{m_l,o}| \otimes |j\>\<i| \Big)|E_{\Bar{e}}\rr \\
        &= \sum_{\Tilde{e},\Bar{e}} u_{e',\Tilde{e}}^* u_{\Bar{e},e} \lambda_{\Tilde{e},\Bar{e},m_l,o} \delta_{j,i} =: \Tilde{\lambda}_{e',e,m_l,o} \delta_{j,i} \;,
    \end{aligned}
    \end{equation}
    for some constant $\Tilde{\lambda}_{e',e,m_l,o} \in \C$.

    For each error sequence $e'$, consider an operator defined by $D_{e'|m_l} = \frac{1}{\sqrt{d_{e',e'}}} \ll F_{e'}|(|C_{m_l}\rr |\Pi_{Q_0}\rr)$.
    Thus by using eqn.~\eqref{eqn:choi_oper_dynamical_KL_condition} and eqn.~\eqref{eqn:transformed_error_opers_KL_cond} the action of $D_{e'|m_l}$ on the codestate at the start of the decoding round is
    \begin{equation}
    \begin{aligned}
        &D_{e'|m_l} (\ll C_{m_l,o}|\<\psi|)|F_e\rr \\
        &= \sum_{i,j} |j\>\, \psi_i  \ll F_{e'}|(|C_{m_l}\rr\ll C_{m_l,o}| \otimes |j\>\<i|)|F_e\rr \\
        &= \sum_{i,j,} |j\>\, \psi_i  \Tilde{\lambda}_{e',e,m_l,o} \delta_{j,i} = \Tilde{\lambda}_{e',e,m_l,o} |\psi\> \;.
    \end{aligned}
    \end{equation}
    Therefore the overall action of a linear map $\mathcal{D}_{m_l}(\rho) = \sum_e D_{e|m_l} \rho D_{e|m_l}^\dag$ on the density operator of the code at the start of the decoding round is
    \begin{equation}
    \begin{aligned}
        &\sum_e \mathcal{D}_{m_l}(|F_e\rr\ll F_e| \ast \mathbf{I}_{m_l} \ast \ketbra{\psi}) \\
        &= \sum_{e,e',o} D_{e'|m_l} (\ll C_{m_l,o}|\<\psi|)|F_e\rr\ll F_e| (|C_{m_l,o}\rr|\psi\>) D_{e'|m_l}^\dag \\
        &= \sum_{e,e',o} \Tilde{\lambda}_{e',e,m_l,o} \Tilde{\lambda}_{e',e,m_l,o}^* \ketbra{\psi} = \lambda_{m_l} \ketbra{\psi}
    \end{aligned}
    \end{equation}
    which recovers the initial state $|\psi\>$ as in~\eqref{eqn:dynamical_QECC}.

    Since $\mathcal{D}_{m_l}$ is a completely positive map, we now show that we can add another operator to make it trace-preserving. 
    Consider polar decomposition
    \begin{equation}
    \begin{aligned}
        &(\ll C_{m_l}| \ll\Pi_{Q_0}|)|F_e\rr \\
        &= U_{e,m_l} \sqrt{\ll F_e| \big( |C_{m_l}\rr\ll C_{m_l}|\otimes| \Pi_{Q_0}\rr\ll\Pi_{Q_0}| \big) |F_e\rr } \\
        &= U_{e,m_l} \Pi_{Q_0}\sqrt{d_{e,e,m_l}}
    \end{aligned}
    \end{equation}
    where the last equality is due to eqn.~\eqref{eqn:diag_constant_dynamical_KL_cond}.
    Then can define projector $\Pi_{e,m_l} = U_{e,{m_l}} \Pi_{Q_0} U_{e,m_l}^\dag = \frac{1}{\sqrt{d_{e,e,m_l}}} (\ll C_{m_l}| \ll\Pi_{Q_0}|)|F_e\rr U_{e,m_l}^\dag$, satisfying orthogonality
    \begin{equation}
    \begin{aligned}
        &\Pi_{e',m_l}^\dag \Pi_{e,m_l} \\
        &= \frac{ U_{e',m_l} \ll F_{e'}|(|C_{m_l}\rr\ll C_{m_l}| \otimes |\Pi_{Q_0}\rr\ll\Pi_{Q_0}|)|F_e\rr U_{e,m_l}^\dag}{\sqrt{d_{e',e',m_l} d_{e,e,m_l}}} \\
        &= \frac{ d_{e',e,m_l} U_{e',m_l} \Pi_{Q_0} U_{e,m_l}^\dag}{\sqrt{d_{e',e',m_l} d_{e,e,m_l}}} \;,
    \end{aligned}
    \end{equation}
    since $d_{e',e,m_l}=0$ for all $e\neq e'$.
    Therefore for each $e$ we have
    \begin{equation}
    \begin{aligned}
        &D_{e|m_l}^\dag D_{e|m_l} = \frac{(\ll C_{m_l}|\ll \Pi_{Q_0}|)|F_e\rr\ll F_e|(|C_{m_l}\rr|\Pi_{Q_0}\rr)}{d_{e,e,m_l}} \\
        &= U_{e,m_l} \Pi_{Q_0} U_{e,m_l}^\dag = \Pi_{e,m_l} \;.
    \end{aligned}
    \end{equation}
    Then by adding projector $\Pi^\perp$ onto a space orthogonal to $\{\Pi_{e,m_l}\}_e$ to the set of operators $\{D_{e|m_l}\}_e$ defining $\mathcal{D}_{m_l}$ we have $\Pi^\perp + \sum_e D_{e|m_l}^\dag D_{e|m_l} = I$, hence $\mathcal{D}_{m_l}$ is trace-preserving.
\end{proof}

As noted in the proof and by using the Kraus operators in eqn.~\eqref{eqn:error_check_kraus_repr}, we can equivalently express eqn.~\eqref{eqn:choi_oper_dynamical_KL_condition} as
\begin{equation}
\begin{aligned}
    &\ll E_{e'}| \Big( |C_{m_l}\rr\ll C_{m_l,o}| \otimes |j\>\<i|\big)|E_e\rr \\
    &= \tr(|K_{e',m_l}\rr\ll K_{e,m_l,o}| \ast |j\>\<i|) \\
    &= \<j|K_{e',m_l}^\dag K_{e,m_l,o}|i\> = \delta_{j,i} \lambda_{e',e,m_l,o} \;,
\end{aligned}
\end{equation}
where $K_{e',m_l} = \sum_o K_{e',m_l,o}$.
This expression tells us for a final memory state $m_l$ of a strategic code correcting $\mathfrak{E}$, the sequence of check measurement outcomes in $O_{m_l}$ forms orthogonal an subspace $\mathscr{V}_{i,m_l}$ for each eigenbasis $\{|i\>_{Q_0}\}_i$ spanning initial codestate $\mathscr{S}_{Q_0}$, regardless of sequence of error.
Namely subspace $\mathscr{V}_{i,m_l}$ is spanned by $\{\ll C_{m_l,o}|\<i|E_e\rr\}_{e,o}$.
Moreover, the independence of constant $\lambda_{e',e,m_l,o}$ from the initial codestate also indicates that the code state at the start of the decoding round is \textit{uncorrelated} with the noise environment, although it generally depends on the final memory state $m_l$.
Due to this independence between the noise environment and the code state, it is sufficient for the decoder to have the information about $m_l$ to recover the initial state, i.e. to construct a projective measurement $\{\Pi_{e,m_l}\}_e$ used in the proof to project the noisy codestate onto subspace $\mathscr{V}_{i,m_l}$ and perform recovery unitary operator $U_{e,m_l}$ according to outcome $e$ to obtain the initial codestate.

For the special case when all check measurement outcomes are stored in the classical memory, i.e. there is a bijection between each memory state $m_r$ and sequence of check measurement outcomes $o_1,\dots,o_r$ for all $r$, the condition in Theorem~\ref{thm:algebraic_KL_condition} can be stated in a more symmetric manner as $\mathbf{I}_{m_l} = |C_{m_l,o}\rr\ll C_{m_l,o}|$.
Since each final memory state $m_l$ and each sequence of check measurement outcomes $o=o_1,\dots,o_l$ have a one-to-one correspondence, we can simply write $|C_{m_l}\rr := |C_{m_l,o}\rr$.

\begin{corollary}\label{cor:algebraic_condition_all_outcome_memory}
    A strategic code $(\mathscr{S}_{Q_0},\mathbf{I})$ storing all check measurement outcomes in its memory corrects error $\mathfrak{E}$ if and only if
    \begin{equation}\label{eqn:symmetric_dynamical_KL_condition}
    \begin{gathered}
        \ll E_{e'} | \big( |C_{m_l}\rr\ll C_{m_l}| \otimes |j\>\<i| \big) | E_{e}\rr = \lambda_{e',e,m_l} \delta_{j,i}
    \end{gathered}
    \end{equation}
    where $\lambda_{e',e,m_l}\in\C$ is some constant for all final memory state $m_l$ and all pairs of error sequences $e,e'$.
\end{corollary}

\subsection{Static quantum error-correction condition as a special case}

From Theorem~\ref{thm:algebraic_KL_condition}, we can recover the Knill-Laflamme necessary and sufficient error-correction condition for static QECC~\cite{knill1997theory} (see also Appendix~\ref{app:KL_condition}), which is
\begin{equation}
    \<j|E_{e'}^\dag E_e|i\> = \lambda_{e',e} \delta_{j,i}
\end{equation}
where $|i\>,|j\>$ is an arbitrary pair of orthonormal vectors spanning codespace $\mathscr{S}_Q$ and $E_e,E_{e'}$ are a pair of Kraus operators of error map $\mathcal{E}(\rho_Q) = \sum_e E_e \rho_Q E_e^\dag$.
The static QECC scenario is obtained by setting the number of rounds to $l=0$ in a general QECC scenario, i.e. there is no sequence check instruments. 
The operator $|C_{m_l}\rr\ll C_{m_l,o}|$ in eqn.~\eqref{eqn:choi_oper_dynamical_KL_condition}simply becomes identity and vectorized error operators are of the form $|E_e\rr = \sum_j E_e|j\>|j\>$ and constant is simply $\lambda_{e',e}$ as there is no dependence on the check measurement outcomes.
Thus eqn.~\eqref{eqn:choi_oper_dynamical_KL_condition} becomes
\begin{equation}
\begin{aligned}
    \lambda_{e',e} \delta_{j,i} &= \ll E_{e'}|j\rr\ll i|E_e\rr = \<j|E_{e'}^\dag E_e|i\>
\end{aligned}
\end{equation}
giving us the Knill-Laflamme static QECC condition.

Lastly we note that without changing the strategic code framework, how strategic code error-correction is defined in Definition~\ref{def:dynamical_QECC} can be generalized as follows.
Instead of requiring the decoder output to be a state proportional to the initial codestate, we can instead introduce additional redundancy by encoding logical information in a subsystem of $\mathscr{S}_{Q_0}$ and requiring the decoder only to recover logical information stored in that subsystem.
Namely, given initial codestate $\rho\otimes\sigma$ we want to recover $\rho\otimes\sigma_{m_l}$ given final memory state $m_l$.
This is the generalization of the subsystem code~\cite{kribs2005unified,kribs2005operator,poulin2005stabilizer,nielsen2007algebraic} to the strategic code framework.
Corresponding to this definition, however, one needs a different necessary and sufficient condition than Theorem~\ref{thm:algebraic_KL_condition} and Theorem~\ref{thm:info_theoretic_condition}, which is left for future work.
For more details on the subsystem code generalization to strategic code see Appendix~\ref{app:strategic_subsystem_code}.

\subsection{Information-theoretic error-correction condition}

For a static QECC (special case of a strategic code with $l=0$), it was shown in~\cite{nielsen1998information} that a necessary and sufficient condition for a completely-positive, trace non-decreasing error map $\mathcal{E}:\mathscr{H}_Q\rightarrow\mathscr{H}_{Q'}$ to be correctable by QECC with codespace $\mathscr{S}_Q$ is
\begin{equation}\label{eqn:static_QECC_info_thheoretic_condition}
\begin{aligned}
    S(\rho^Q) = S(\rho^{Q'}) - S(\rho^{E'}) = S(\rho^{Q'}) - S(\rho^{R'Q'}) \;.
\end{aligned}
\end{equation}
Here a reference system $R$ is introduced, and $\rho^Q=\tr_R(\ketbra{\phi}^{RQ})$ where $\ketbra{\phi}^{RQ}$ is the maximally entangled state between initial system $Q$ and reference system $R$.
So, we have the density operators after the error $\rho^{Q'} = \mathcal{E}(\rho^Q) / \tr(\mathcal{E}(\rho^Q))$ and $\rho^{R'Q'} = \mathcal{I}_R\otimes\mathcal{E}(\ketbra{\phi}^{RQ})$ where $\rho^{E'}$ the marginal state of the noise environment of $\mathcal{E}$ when expressed as 
\begin{equation}
    \mathcal{E}(\cdot) = \tr_{E'}\Big( (I_{Q'}\otimes\Pi) V \cdot V^\dag(I_{Q'}\otimes\Pi) \Big)
\end{equation}
for some isometry $V:\mathscr{S}_Q\rightarrow\mathscr{S}_{Q'}\otimes\mathscr{S}_{E'}$ and orthogonal projector $\Pi\in\mathscr{H}_{E'}$.
Generalization of information-theoretic condition~\eqref{eqn:static_QECC_info_thheoretic_condition} to subsystem codes is shown in~\cite{nielsen2007algebraic}.
The term $S(\rho^{Q'}) - S(\rho^{R'Q'})$ is the so-called ``coherent information'', which quantifies the amount of information about $\rho^Q$ contained in $\rho^{Q'}$~\cite{schumacher1996quantum,nielsen1998information,nielsen2007algebraic}.

In a general QECC scenario, we instead have a sequence of completely positive map $\mathcal{E}^{(0)},\dots,\mathcal{E}^{(l)}$. 
In what follows, we omit normalization for the states and density operators for notational simplicity.
Now consider the density operator in $\mathscr{H}_{R_l'}\otimes\mathscr{H}_{Q_l'}\otimes\mathscr{H}_{M_l}\otimes\mathscr{H}_{E_l}$ with one-half of a maximally entangled state $\sum_i |i\>_{R_0} |i\>_{Q_0}$ as an input initial state in $\mathscr{S}_{Q_0}$ and given final memory state $m_l$
\begin{widetext}
\begin{equation}\label{eqn:total_final_density_given_memory_and_error}
\begin{aligned}
    \rho_{m_l,e,e',o,o'}^{R_l'Q_l'M_lE_l} 
    &= \sum_{i,j} |i\>\<j|_{R_0} \otimes \big( K_{e,m_l,o} |i\>\<j|_{Q_0} K_{e',m_l,o'}^\dag \big) \otimes |o\>\<o'|_{M_l} \otimes |e\>\<e'|_{E_l}
\end{aligned}
\end{equation}
\end{widetext}
where $o,o'\in O_{m_l}$ is a pair of sequences of check measurement outcomes resulting in final memory state $m_l$ and $K_{e,m_l,o}$ is an operator defined by $|K_{e,m_l,o}\rr = \ll C_{m_l,o}|E_e\rr$ (see eqn.~\eqref{eqn:error_check_kraus_repr}. 
This scenario is illustrated in Fig.~\ref{fig:dynamical_QECC2}.

\begin{figure}
    \centering
    \includegraphics{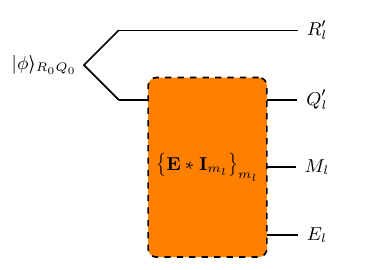}
    \caption{QECC scenario where the initial code state is one-half of maximally entangled state $|\phi\>=\sum_i |i\>_{R_0}|i\>_{Q_0}$ between a reference $R_0$ system and the initial code space $Q_0$.
    After $l$ rounds of errors and check measurements ending in final memory state $m_l$, we obtain the global density oeprator $\rho_{m_l}^{R_l'Q_l'M_lE_l}$ of the reference system, code system, check measurement outcome sequence, and noise environment.}
    \label{fig:dynamical_QECC2}
\end{figure}

Now consider density operator
\begin{equation}\label{eqn:total_final_density_given_memory}
\begin{aligned}
    \rho_{m_l}^{R_l'Q_l'M_lE_l} &:= \sum_{e,e',o,o'} \rho_{m_l,e,e',o,o'}^{R_l'Q_l'M_lE_l} \;.
\end{aligned}
\end{equation}
We also define the marginal density operators as $\rho_{m_l}^{R_l'} := \tr_{Q_l'M_lE_l}(\rho_{m_l}^{R_l'Q_l'M_lE_l})$ and $\rho_{m_l}^{M_lE_l} := \tr_{R_l'Q_l'}(\rho_{m_l}^{R_l'Q_l'M_lE_l})$.
We also denote the density operator over all possible final memory state $m_l$ as
\begin{equation}
\begin{aligned}
    \rho^{R_l'Q_l'M_lE_l} 
    &= \sum_{m_l} P_{M_l}(m_l) \Tilde{\rho}_{m_l}^{R_l'Q_l'M_lE_l}
\end{aligned}
\end{equation}
for $P_{M_l}(m_l) = \tr(\rho_{m_l}^{R_l'Q_l'M_lE_l})$ and density operator $\Tilde{\rho}_{m_l}^{R_l'Q_l'M_lE_l} = \rho_{m_l}^{R_l'Q_l'M_lE_l}/P_{M_l}(m_l)$.
Here, $P_{M_l}(m_l)$ can be interpreted as the probability of the final memory state being $m_l$.
Hence $\Tilde{\rho}_{m_l}^{R_l'Q_l'M_lE_l}$ is the density operator at the start of the decoding round, given that the memory storing information about the check measurement outcomes is $m_l$.

Now we show a necessary and sufficient information-theoretic conditions for strategic code to correct error $\mathfrak{E}=\{|E_e\rr\ll E_e|\}_{e}$.

\begin{theorem}\label{thm:info_theoretic_condition}
    The following are equivalent:
    \begin{enumerate}
        \item A strategic code $(\mathscr{S}_{Q_0},\mathbf{I})$ corrects error $\mathfrak{E}$.

        \item $S(\rho_{m_l}^{R_l'M_lE_l}) = S(\rho_{m_l}^{R_l'}) + S(\rho_{m_l}^{M_lE_l})$ for all final memory state $m_l$ such that $P_{M_l}(m_l) > 0$.

        \item $I_{\rho_{m_l}^{R_l'M_lE_l}}(R_l' : M_lE_l) = 0$ for all final memory state $m_l$ such that $P_{M_l}(m_l) > 0$.
    \end{enumerate}
\end{theorem}

\begin{proof}
    First we use Theorem~\ref{thm:algebraic_KL_condition} to show that eqn.~\eqref{eqn:choi_oper_dynamical_KL_condition} implies $S(\rho_{m_l}^{R_l'M_lE_l}) = S(\rho_{m_l}^{R_l'}) + S(\rho_{m_l}^{M_lE_l})$.
    The merginal density operator in $\mathscr{H}_{R_l'}\otimes\mathscr{H}_{M_l}\otimes\mathscr{H}_{E_l}$ can be expressed as
    \begin{equation}\label{eqn:reduced_RME_density operator}
    \begin{aligned}
        &\rho_{m_l}^{R_l'M_lE_l} := \sum_{e,e',o,o'} \tr_{Q_l'}\big( \rho_{m_l,e,e',o,o'}^{R_l'Q_l'M_lE_l} \big) \\
        &= \sum_{e,e',o,o',i,j} |i\>\<j|_{R_l'} \<j|K_{e',m_l,o'}^\dag K_{e,m_l,o}|i\> \\
        &\quad \otimes |o\>\<o'|_{M_l} \otimes |e\>\<e'|_{E_l} \;.
    \end{aligned}
    \end{equation}

    Now consider a unitary $[u_{\Bar{e},e}]_{\Bar{e},e}$ as defined in the sufficiency of eqn.~\eqref{eqn:choi_oper_dynamical_KL_condition} which performs the transformation $|F_e\rr = \sum_{\Bar{e}} u_{\Bar{e},e} |E_{\Bar{e}}\rr$.
    Applying this to eqn.~\eqref{eqn:reduced_RME_density operator} transforms the Kraus operator $K_{e,m_l,o} \mapsto F_{e,m_l,o}$ and the noise environment basis $|e\> \mapsto |v_e^{(m_l)}\> = \sum_{\Bar{e}} u_{\Bar{e},e} |\Bar{e}\>$.
    Also consider the decoding channel $\mathcal{D}_{m_l}$ with Kraus operators $\{D_{e|m_l}\}_e$ constructed in proof of the sufficiency of eqn.~\eqref{eqn:choi_oper_dynamical_KL_condition} to correct $\mathfrak{E}$.
    As $\sum_e D_{e|m_l}^\dag D_{e|m_l} = I$, we obtain
    \begin{equation}
    \begin{aligned}
        &\rho_{m_l}^{R_l'M_lE_l} \\
        &= \sum_{e,e',o,o',i,j,\Tilde{e}} |i\>\<j|_{R_l'} \<j|F_{e',m_l,o'}^\dag D_{\Tilde{e}|m_l}^\dag D_{\Tilde{e}|m_l} F_{e,m_l,o}|i\> \\
        &\quad \otimes |o\>\<o'|_{M_l} \otimes |v_e^{(m_l)}\>\<v_{e'}^{(m_l)}|_{E_l} \\
        &= \Pi_{Q_0} \otimes \Big( \sum_{e,e',o,o',\Tilde{e}} \Tilde{\lambda}_{\Tilde{e},e',m_l,o'}^* \Tilde{\lambda}_{\Tilde{e},e,m_l,o} |o\>\<o'|_{M_l} \\
        &\quad \otimes |v_e^{(m_l)}\>\<v_{e'}^{(m_l)}|_{E_l} \Big)
    \end{aligned}
    \end{equation}
    since $D_{\Tilde{e}|m_l} F_{e,m_l,o}|i\> = \Tilde{\lambda}_{\Tilde{e},e,m_l,o}|i\>$ for some constant $\Tilde{\lambda}_{\Tilde{e},e,m_l,o}\in\C$.
    Thus $\rho_{m_l}^{R_l'M_lE_l} = \rho_{m_l}^{R_l'} \otimes \rho_{m_l}^{M_lE_l}$, which is equivalent to $S(\rho_{m_l}^{R_l'M_lE_l}) = S(\rho_{m_l}^{R_l'}) + S(\rho_{m_l}^{M_lE_l})$.

    Now we show that $S(\rho_{m_l}^{R_l'M_lE_l}) = S(\rho_{m_l}^{R_l'}) + S(\rho_{m_l}^{M_lE_l})$ implies eqn.~\eqref{eqn:dynamical_QECC} by constructing a decoding channel recovering the initial code state.
    Now we consider a Schmidt decomposition on bipartition between codespace $Q_l'$ and joint system $R_l'M_lE_l$ of $|\varphi_{m_l}\> = \sum_{i,e,o}|i\>_{R_0} (K_{m_l,o,e}|i\>_{Q_0}) |o\>_{M_l} |e\>_{E_l}$ for each $o$, which gives
    \begin{equation}
    \begin{aligned}
        |\varphi_{m_l}\> = \sum_{i,\alpha} \sqrt{q_\alpha^{(m_l)}} |i\>_{R_l'} |u_\alpha^{(m_l)}\>_{M_lE_l} |v_{i,\alpha}^{(m_l)}\>_{Q_l'}
    \end{aligned}
    \end{equation}
    where $\{|u_\alpha^{(m_l)}\>_{M_lE_l}\}_\alpha$ is an eigenvector of $\rho_{m_l}^{M_lE_l}$ with corresponding eigenvalue $q_\alpha^{(m_l)}$ and $\{|v_{i,\alpha}^{(m_l)}\>\}_{i,\alpha}$ is a set of orthonormal vectors in code space $\mathscr{S}_{Q_l'}$. 

    Now consider decoding channel $\mathcal{D}_{m_l}$ with Kraus operators $\{D_{\alpha|m_l}\}_{\alpha} \cup \{\Pi^\perp\}$ defined by 
    \begin{equation}
        D_{\alpha|m_l} = V_{m_l,\alpha} \sum_i\ketbra{v_{i,\alpha}^{(m_l)}}_{Q_l'} \;,
    \end{equation}
    for unitary $V_{m_l,\alpha}:\mathscr{S}_{Q_l'}\rightarrow\mathscr{S}_{Q_0}$ such that $V_{m_l,\alpha}|v_{i,\alpha}^{(m_l)}\>_{Q_l'} = |i\>_{Q_0}$ for all $i$.
    Operator $\Pi^\perp$ is a projector onto subspace $\mathscr{V}^\perp$ orthogonal to $\spn\{|v_{i,\alpha}^{(m_l)}\>\}_\alpha$ to obtain normalization $\Pi^\perp + \sum_\alpha D_{\alpha|m_l}^\dag D_{\alpha|m_l} =I$.
    Hence
    \begin{equation}\label{eqn:decoding_info_theoretic_cond}
    \begin{aligned}
        D_{\alpha|m_l}|\varphi_{m_l}\> &= \sum_i \sqrt{q_{\alpha}^{(m_l)}} |i\>_{R_l'} |u_\alpha^{(m_l)}\>_{M_lE_l} |i\>_{Q_0} \\
        &= |\phi\>_{R_l'Q_l'} \sqrt{q_{\alpha}^{(m_l)}} |u_\alpha^{(m_l)}\>_{M_lE_l}
    \end{aligned}
    \end{equation}
    showing that the initial maximally entangled state is recovered.
    Thus the overall action of decoding channel $\mathcal{D}_{m_l}$ to the density operator $|E_e\rr\ll E_e| \ast \mathbf{I}_{m_l} \ast \ketbra{\psi}_{Q_0}$ of system $Q_l'$ at the start of the decoding round is 
    \begin{equation}
    \begin{aligned}
        &\mathcal{D}_{m_l}(\mathbf{E} \ast \mathbf{I}_{m_l} \ast \ketbra{\psi}_{Q_0}) \\
        &= \sum_{i,j,e} \psi_i\psi_j^* \mathcal{D}_{m_l}(|E_e\rr\ll E_e| \ast \mathbf{I}_{m_l} \ast |i\>\<j|_{Q_0}) \\
        &= \sum_{i,j,e',e,o',o} \psi_i\psi_j^* \mathcal{D}_{m_l}(K_{e,m_l,o} |i\>\<j| K_{e',m_l,o'}^\dag) \<e'|e\>\<o'|o\> \\
        &= \sum_{i,j,\alpha',\alpha} \psi_i\psi_j^* \sqrt{q_\alpha^{(m_l)}q_{\alpha'}^{(m_l)}} \mathcal{D}_{m_l}(|v_{i,\alpha}^{(m_l)}\>\<v_{j,\alpha'}^{(m_l)}|) \\
        &\quad \times \<u_{\alpha'}^{(m_l)}|u_\alpha^{(m_l)}\> \\
        &= \sum_{i,j,\alpha} \psi_i\psi_j^* q_\alpha^{(m_l)} \mathcal{D}_{m_l}(|v_{i,\alpha}^{(m_l)}\>\<v_{j,\alpha}^{(m_l)}|) \\
        &= \sum_{i,j,\alpha} \psi_i\psi_j^* q_{\alpha}^{(m_l)} |i\>\<j|_{Q_0} \\
        &= \lambda_{m_l} \ketbra{\psi}_{Q_0}
    \end{aligned}
    \end{equation}
    where $\lambda_{m_l} = \sum_\alpha q_\alpha^{(m_l)}$.
    To obtain the third equality, we use the change of basis on the joint memory - noise environment system $M_lE_l$ to $|u_\alpha^{(m_l)}\>_{M_lE_l} = \sum_{o,e} \eta_{(o,e),\alpha}^{(m_l)} |o\>_{M_l}|e\>_{E_l}$ for some complex numbers $\{\eta_{(o,e),\alpha}^{(m_l)}\}_{o,e,\alpha}$ which also gives $K_{e,m_l,o}|i\> \mapsto \sqrt{q_\alpha^{(m_l)}}|u_{\alpha,i}^{(m_l)}\> = \sum_{o,e} \eta_{(o,e),\alpha}^{(m_l)} K_{e,m_l,o}|i\>$.
    Whereas the fourth equality is obtained by using eqn.~\eqref{eqn:decoding_info_theoretic_cond}.
    Thus decoding channel $\mathcal{D}_{m_l}$ recovers all initial codestate $|\psi\>_{Q_0}$ for any error sequence $e$.
    
    Lastly, we show that $I_{\rho_{m_l}^{R_l'M_lE_l}}(R_l' : M_lE_l) = 0$ if and only if $S(\rho_{m_l}^{R_l'M_lE_l}) = S(\rho_{m_l}^{R_l'}) + S(\rho_{m_l}^{M_lE_l})$.
    This simply follows from the definition of von Neumann mutual information
    \begin{equation}
    \begin{aligned}
        I_{\rho_{m_l}^{R_l'M_lE_l}}(R_l' : M_lE_l) = S(\rho_{m_l}^{R_l'}) + S(\rho_{m_l}^{M_lE_l}) - S(\rho_{m_l}^{R_l'M_lE_l}) \;,
    \end{aligned}
    \end{equation}
    which is equal to $0$ if and only if $S(\rho_{m_l}^{R_l'M_lE_l}) = S(\rho_{m_l}^{R_l'}) + S(\rho_{m_l}^{M_lE_l})$.
\end{proof}

Statement 2 of Theorem~\ref{thm:info_theoretic_condition} reduces to the information-theoretic necessary and sufficient condition for static QECC in~\cite{schumacher1996quantum,nielsen1998information} stating that the reduced density operator of the reference system and the environment after the error operation $\rho^{R'E'}$ is separable.
Stated equivalently in terms of the von Neumann entropy, $S(\rho^{R'E'}) = S(\rho^{R'}) + S(\rho^{E'})$.
On the other hand, statement 3 of Theorem~\ref{thm:info_theoretic_condition} reduces to the necessary and sufficient condition in~\cite{cerf1997information} for static error correction, stating that it must hold that the mutual information between the noise environment $E'$ and the reference system $R'$ after error operation is $I_{\rho^{R'E'}}(E':R')=0$.
Namely, there is no correlation between $E'$ and $R'$.
In the general QECC case, the condition $S(\rho_{m_l}^{R_l'M_lE_l}) = S(\rho_{m_l}^{R_l'}) + S(\rho_{m_l}^{M_lE_l})$ and $I_{\rho_{m_l}^{R_l'M_lE_l}}(R_l' : M_lE_l) = 0$ indicates that the reference system $R_l'$ and the joint check measurement outcome - noise environment system $M_lE_l$ at the start of the decoding round are uncorrelated for each final memory state $m_l$.
However in general, the check measurement outcome system $M_l$ and the noise environment $E_l$ exhibit some correlation.

\section{Error-adapted Approximate Strategic Code}

So far we have been focusing on strategic codes which recovers logical information \textit{exactly} by showing necessary and sufficient conditions of how to achieve this (Theorem~\ref{thm:algebraic_KL_condition} and Theorem~\ref{thm:info_theoretic_condition}).
However in practice, resource limitations and some knowledge about characteristic of the relevant noise often allow us to relax the requirements on how well logical information should be recovered in exchange for a less resource-intensive code.
These practical considerations gives rise to \textit{approximate} (static) QECCs, which have been known to achieve a performance comparable to generic exact QECC in dealing with a particular error model in a more resource-efficient manner (see e.g.~\cite{leung1997approximate,lidar2013quantum,crepeau2005approximate,fletcher2007optimum,ng2010simple,cafaro2014approximate}).

To address this practical considerations we turn to \textit{approximate strategic code}, namely one where we demand that logical information is recovered only up to a certain fidelity by a decoding channel after the $l$ rounds of operation by the interrogator.
To obtain this approximate code, we propose an optimization problem that given an ensemble of $d'$-dimensional quantum states, an $l$-rounds of error $\mathbf{E}$, and positive integer $d>d'$ and returns: (1) an encoding channel $\mathcal{C}^{(0)}$ mapping bounded linear operators on $\C^{d'}$ to bounded linear operators on subspace $\mathscr{S}_{Q_0}$ of $\C^d$, (2) an $l$-round interrogator $\mathbf{I}$, and (3) set of decoders $\mathbf{D}$ corresponding to each final memory state of interrogator $\mathbf{I}$.
Note that as opposed to the previously considered QECC scenario where we start with the codespace $\mathscr{S}_{Q_0}$, here we start with a channel $\mathcal{C}^{(0)}$ that maps $d'$-dimensional density operators on $\C^{d'}$ to density operators with support on $\mathscr{S}_{Q_0}\subseteq\C^d$.
Also, the decoding channel $\mathcal{D}_{m_l}$ for final memory state $m_l$ maps density operators with support on $\mathscr{S}_{Q_l'}\subseteq\C^d$ (the codespace after the final error map) to density operators on $\C^{d'}$.
Let us denote the space of operators on $\C^{d'}$ at the input of the encoding channel by $\mathscr{H}_L$ and those at the output of the decoding channel by $\mathscr{H}_{L'}$.

We can describe the encoding, interrogator, and decoding as one single quantum comb
\begin{equation}
\begin{aligned}
    \mathbf{Q} = \sum_{m_{0:l}} \mathbf{D}_{m_l} \otimes \mathbf{C}_{m_l|m_{l-1}}^{(l)} \otimes\dots\otimes \mathbf{C}_{m_1}^{(1)} \otimes \mathbf{C}^{(0)} \;.
\end{aligned}
\end{equation}
Operator $\mathbf{Q}$ is positive semidefinite as a consequence of each $\mathbf{D}_{m_l}, \mathbf{C}_{m_l|m_{l-1}}^{(l)} ,\dots, \mathbf{C}_{m_1}^{(1)} , \mathbf{C}^{(0)}$ (Choi operators of CP maps $\mathcal{D}_{m_l},\mathcal{C}_{m_l|m_{l-1}}^{(l)} ,\dots, \mathcal{C}_{m_1}^{(1)} , \mathcal{C}^{(0)}$) being positive semidefinite operator in $\mathscr{H}_{L'}\otimes\mathscr{H}_{Q_l'} \,,\, \mathscr{H}_{Q_l}\otimes\mathscr{H}_{Q_{l-1}'} \,,\dots,\, \mathscr{H}_{Q_1}\otimes\mathscr{H}_{Q_0'} \,,\, \mathscr{H}_{Q_0}\otimes\mathscr{H}_L$, respectively.
The sequence of errors described similarly as $\mathbf{E} = \sum_e |E_e\rr\ll E_e|$, which is also positive semidefinite.

Let $Q = \{L,Q_0,Q_0',\dots,Q_l,Q_l',L'\}$ be the set of labels of the code spaces in the dynamical code and for $\Tilde{Q}\subseteq Q$ denote $(\cdot)^{\top_{\Tilde{Q}}}$ as partial transpose over spaces with labels in $\Tilde{Q}$ and $\tr_{\Tilde{Q}}$ as partial trace over spaces with labels in $\Tilde{Q}$ and $I_{\Tilde{Q}} = \bigotimes_{Q'\in\Tilde{Q}} I_{Q'}$.
Consider a channel $\mathcal{T}:\mathscr{H}_L\rightarrow\mathscr{H}_{L'}$ composed of the sequence of the check instruments $\{\mathcal{C}_{m_l|m_{l-1}}^{(l)},\dots,\mathcal{C}_{m_1}^{(1)},\mathcal{C}^{(0)}\}_m$ for sequence of memory state $m=m_1,\dots,m_{l}$, error maps $\mathcal{E}^{(l)},\dots,\mathcal{E}^{(0)}$ and the final decoding channels $\{\mathcal{D}_{m_l}\}_{m_l}$. 
We use the entanglement fidelity of channel $\mathcal{T}$ on state $\rho$ as our performance metric, which is defined as
\begin{equation}\label{eqn:entanglement_fidelity}
\begin{aligned}
    F(\rho,\mathcal{T})  
    &= \tr\big( \mathbf{E}\ast\mathbf{Q} \, |\rho\rr\ll\rho| \big) \\
    &= \tr\big( (\mathbf{E}^\top \otimes I_{L',L}) \, \mathbf{Q} \, (|\rho\rr\ll\rho|\otimes I_{Q\backslash L,L'}) \big) \;,
\end{aligned}
\end{equation}
where $|\rho\rr = \sum_j \rho|j\>|j\>$ is the vectorized form of density operator $\rho$.

For $r\in\{1,\dots,l\}$, it holds that $\tr_{Q_r}(\sum_{m_r}\mathbf{C}_{m_r|m_{r-1}}^{(r)}) = I_{Q_{r-1}'}$ since $\sum_{m_r} \mathcal{C}_{m_r|m_{r-1}}^{(r)}$ is a CPTP map.
Similarly it also holds that $\tr_{Q_0}(\mathbf{C}^{(0)}) = I_L$ and $\tr_{L'}(\mathbf{D}_{m_l}) = I_{Q_l'}$ for each $m_l$.
Thus for a given error operator $\mathbf{E}$ and initial state $\rho$, we can maximize entanglement fidelity~\eqref{eqn:entanglement_fidelity} over $\mathbf{Q}$ as
\begin{equation}\label{eqn:max_entanglement_fidelity1}
\begin{gathered}
    \max_{\mathbf{Q}} \tr\big( (\mathbf{E}^\top \otimes I_{L',L}) \, \mathbf{Q} \, (|\rho\rr\ll\rho|\otimes I_{Q\backslash L,L'}) \big) \\
    \textup{such that} \\
    \mathbf{Q} = \sum_{m_{0:l}} \mathbf{D}_{m_l} \otimes \mathbf{C}_{m_l|m_{l-1}}^{(l)} \otimes\dots\otimes \mathbf{C}_{m_0}^{(0)} \\
    \mathbf{D}_{m_l} \geq 0 \;,\quad \tr_{L'}(\mathbf{D}_{m_l}) = I_{Q_l'} \\
    \mathbf{C}^{(0)} \geq0 \;,\quad \tr_{Q_0}\big(\mathbf{C}^{(0)}\big) = I_L \\
    \mathbf{C}_{m_r|m_{r-1}}^{(r)} \geq0 \;,\; \tr_{Q_r}\big(\sum_{m_r} \mathbf{C}_{m_r|m_{r-1}}^{(r)}\big) = I_{Q_{r-1}'} \;,\; \forall r\geq1 \;.
\end{gathered}
\end{equation}

We can also impose this normalization condition to operator $\mathbf{Q}$ as follows.
Let $\mathbf{Q}_{m_l} = \sum_{m_{0:l-1}} \mathbf{D}_{m_l} \otimes \mathbf{C}_{m_l|m_{l-1}}^{(l)} \otimes\dots\otimes \mathbf{C}_{m_1}^{(1)} \otimes \mathbf{C}^{(0)}$ (hence $\mathbf{Q}=\sum_{m_l}\mathbf{Q}_{m_l}$) and for $r\in\{0,\dots,l\}$ let $\mathbf{Q}_{m_r}^{(r)} = \sum_{m_{0:r-1}} \mathbf{C}_{m_r|m_{r-1}}^{(r)} \otimes\dots\otimes \mathbf{C}_{m_1}^{(1)} \otimes \mathbf{C}^{(0)}$ (hence $\mathbf{Q}_{m_1}^{(1)}=\mathbf{C}_{m_1}^{(1)} \otimes \mathbf{C}^{(0)}$ and $\mathbf{Q}^{(0)}=\mathbf{C}^{(0)}$).
Thus we can rewrite the conditions in~\eqref{eqn:max_entanglement_fidelity1} in this notation as
\begin{equation} \label{eqn:cone_program}
\begin{gathered}
    \max_{\mathbf{Q}} \tr\big( (\mathbf{E}^\top \otimes I_{L',L}) \, \mathbf{Q} \, (|\rho\rr\ll\rho|\otimes I_{Q\backslash L,L'}) \big) \\
    \textup{such that} \\
    \mathbf{Q}\geq0 \;,\quad \tr_{L'}(\mathbf{Q}) = \sum_{m_l} I_{Q_l'} \otimes \mathbf{Q}_{m_l}^{(l)} \\
    \tr_{Q_r}\Big( \sum_{m_r} \mathbf{Q}_{m_r}^{(r)} \Big) = \sum_{m_{r-1}} I_{Q_{r-1}'} \otimes  \mathbf{Q}_{m_{r-1}|m_{r-2}}^{(r-1)} \;,\,\forall r\geq1 \\
    \mathbf{Q}^{(0)}\geq0 \;,\quad \tr_{Q_0}(\mathbf{Q}^{(0)}) = I_L \\
    \mathbf{Q}_{m_r|m_{r-1}}^{(r)} \geq 0 \;,\quad \forall r\geq1 \;.
\end{gathered}
\end{equation}

The optimization problem in \eqref{eqn:cone_program} is an instance of conic programming~\cite{boyd2004convex,nemirovski2006advances}, where the cone characterized by $\mathbf{Q}$ is a separable cone. The aforementioned conic programming can be solved using see-saw algorithm, where every iteration of the see-saw is a semidefinite program (SDP)~\cite{boyd2004convex}. In the special case of static QECC, our conic program in \eqref{eqn:cone_program} reduces to the  bi-convex optimization structure from Ref.~\cite{fletcher2007optimum} and can be solved using two SDPs running one after another, until convergence within a fixed tolerance is attained.

\section{Discussions}

In this work, we propose a unified framework for quantum error-correcting codes (QECC) called the strategic code.
It encompasses all existing QECCs and all physically plausible QECCs to be discovered, including codes involving operational adaptivity and also considering effects of spatially and temporally (non-Markovian) correlated error models.
The strategic code introduces a device called an interrogator which represents all operations performed by the coder in between encoding and decoding stages.
The interrogator is general, in that it may contain any set of operations performed both spatially or temporally (in sequence) with classical or quantum memory.
Within this framework we show an algebraic (Theorem~\ref{thm:algebraic_KL_condition} and an information-theoretic (Theorem~\ref{thm:info_theoretic_condition}) necessary and sufficient error-correction conditions.
These conditions apply to all known variants of dynamical QECC (and all physically-allowed generalizations) and include the error-correction conditions for static QECC~\cite{knill1997theory,nielsen1998information,cerf1997information} as a special case.
The generality of the results partly owes to the quantum combs formalism, which gives a natural spatio-temporal representation for a QECC, as it has been for many sequential tasks in quantum information and computation.
In this formalism, we also propose an optimization problem that gives an approximate QECC that recovers logical information up to desired level of fidelity for a given error model, which again may exhibit non-Markovian correlations.

As mentioned in the main text, although we focus on the scenario where the interrogator only maintains classical memory, the strategic code also accommodates an interrogator with quantum memory (as discussed in detail in Appendix~\ref{app:interrogator_quantum_memory}).
This leads to many questions including: How does the size (dimension) quantum quantum memory affects error-correction? What are the necessary and sufficient conditions for error correction given a fixed size quantum memory?
Moreover, as this generalization includes the entanglement-assisted QECC (EAQECC)~\cite{brun2006correcting,hsieh2007general,brun2014catalytic} as a special case, one could investigate into relationships between error-correction conditions for EAQECC (e.g.~\cite{grassl2022entropic}) to analogous conditions for strategic code.
Also, another generalization mentioned in the main text to allow encoding of logical information in a susbsystem of the codespace, analogous to subsystem codes (as discussed in detail in Appendix~\ref{app:strategic_subsystem_code}).
As subsystem codes has a different error-correction conditions~\cite{poulin2005stabilizer,kribs2005unified,nielsen2007algebraic} compared to traditional static codes~\cite{knill1997theory,nielsen1998information,cerf1997information}, it is an interesting future work to show necessary and sufficient error-correction condition of subsystem strategic code.
Another interesting future work is to use the concept of quantum combs virtualization~\cite{takagi2024virtual,yuan2024virtual} to the strategic code.
This is essentially a method of approximating some operator $\mathbf{\Phi}$ which involves randomly choosing from a set of ``allowed'' $l$-rounds strategic codes $\{\mathbf{I}^{(k)}\}_k$ followed by a post-processing.
The sampling procedure and post-processing are constructed based on a linear expansion of some operator $\mathbf{\Phi}$ in terms of $\{\mathbf{I}^{(k)}\}_k$.
Here operator $\mathbf{\Phi}$ have the same dimension as $\mathbf{I}^{(k)}$, but it may correspond to a non-physical process, such as those involving indefinite causal order or causally inseparable~\cite{oreshkov2012quantum,chiribella2013quantum,procopio2015experimental,costa2016quantum,milz2018entanglement,rubino2017experimental,goswami2018indefinite,loizeau2020channel,ebler2018enhanced}.
One could also explore how a strategic code equipped with such exotic causal structure performs.
For more detailed discussion on strategic code virtualization and strategic code with more exotic causal structures, see Appendix~\ref{app:virtual_strategic_code}.

Further work could be done on an explicit construction of dynamical QECCs such as a Floquet code, that corrects a sequence of error maps by using the strategic code framework and conditions in Theorem~\ref{thm:algebraic_KL_condition} and Theorem~\ref{thm:info_theoretic_condition}.
It would also be an interesting to explore further whether adaptive strategic code can provide any advantage over codes with fixed sequence of operations.
Such advantage could take the form of larger code distance or capability of storing more logical information.
A notion of approximate strategic code can also be explored further and optimized using our optimization method.
Particularly, one could perhaps show an approximate error-correction condition for strategic code with respect to logical information recovery up to a certain fidelity, analogous to approximate static QECC conditions in~\cite{beny2009general}.
It is also interesting to understand further the relationship between our conditions and the operator algebraic condition in~\cite{gottesman2022opportunities,fu2024error} which is formulated in terms of non-Abelian gauge group defined by the sequence of Clifford gates and Pauli measurements, as well as relationship between the strategic code and other QECC frameworks~\cite{bombin2023unifying,kesselring2024anyon,davydova2023quantum,fu2024error}.

\hfill

\section*{Acknowledgements}

This work is supported by the NRF2021-QEP2-02-P06 from the Singapore Research Foundation, the Singapore Ministry of Education Tier 1 Grant RG77/22 (S), the FQXi R-710-000-146-720
Grant “Are quantum agents more energetically efficient
at making predictions?” from the Foundational Questions Institute, Fetzer Franklin Fund (a donor-advised
fund of Silicon Valley Community Foundation) and A*STAR C230917003.
AT is supported by CQT PhD scholarship. The authors thank Tobias Haug,  Varun Narsimhachar and Yunlong Xiao for interesting discussions.

\bibliographystyle{unsrt}
\bibliography{references}

\onecolumngrid
\appendix

\section{Necessary and sufficient algebraic conditions for static QECC}\label{app:KL_condition}

Knill-Laflamme's necessary and sufficient condition for exact QECC ~\cite{knill1997theory} states that for a given basis $\{|i\>_Q\}_i$ of code space $\mathscr{S}_Q$ and any distinct pair of code space basis $\ket{i}_Q,\ket{j}_Q\in\mathscr{S}_Q$, it holds that
\begin{equation}
\begin{gathered}
    \<i|_QE_a^\dag E_b|i\>_Q = \<j|_QE_a^\dag E_b|j\>_Q = \lambda_{a,b} \\
    \quad \textup{and} \quad \\
    \<i|_QE_a^\dag E_b|j\>_Q = 0 \;,
\end{gathered}
\end{equation}
for some constant $\lambda_{a,b}\in\C$.
Equivalently for a projector $\Pi_Q = \sum_i \ketbra{i}_Q$ onto codespace $\mathscr{S}_Q$, it holds that
\begin{equation}
    \Pi_Q E_a^\dag E_b \Pi_Q = \lambda_{a,b} \Pi_Q \;.
\end{equation}

For subsystem QECC with code space $\mathscr{S}_Q = \mathscr{S}_C \otimes \mathscr{S}_G$, where $\mathscr{S}_C$ is the code subsystem and $\mathscr{S}_G$ is the gauge subsystem, Nielsen-Poulin's necessary and sufficient condition for exact QECC~\cite{nielsen2007algebraic} is
\begin{equation}
    \Pi_Q E_a^\dag E_b \Pi_Q = I_C \otimes g_{a,b}
\end{equation}
where $g_{a,b}$ is an operator on $\mathscr{S}_G$ and $\Pi_Q$ projection onto code space $\mathscr{S}_Q$ defined as $\Pi=VV^\dag$ where $V:\mathscr{S}_L\rightarrow\mathscr{S}_Q$ is an isometry that encodes logical states into code states.

\section{Quantum combs representation of strategic code}\label{app:choi_jamiolkowski_isomorphism_link_product}

Since $\mathcal{E}^{(r)}$ is a CP map and $\sum_{m_r} \mathcal{C}_{m_r|m_{r-1}}^{(r)}$ is a CPTP map their Choi operators $\mathbf{E}^{(r)}$ and $\mathbf{C}_{m_r|m_{r-1}}^{(r)}$ are positive definite.
Hence $\mathbf{E}^{(r)}$ and $\mathbf{C}_{m_r|m_{r-1}}^{(r)}$ admits decomposition
\begin{equation}
\begin{gathered}
    \mathbf{E}^{(r)} = \sum_{e_r} |E_{e_r}\rr\ll E_{e_r}| \\
    \mathbf{C}_{m_r|m_{r-1}}^{(r)} = \sum_{o_r:f(o_r,m_{r-1})=m_r} |C_{o_r|m_{r-1}}\rr\ll C_{o_r|m_{r-1}}|
\end{gathered}
\end{equation}
where $|E_{e_r}\rr = \sum_{i,j} \<i|E_{e_r}|j\> |i\>|j\>$ and $|C_{o_r|m_{r-1}}\rr = \sum_{i,j} \<i|C_{o_r|m_{r-1}}|j\> |i\>|j\>$ are the (unnormalized) eigenvectors of $\mathbf{E}^{(r)}$ and $\sum_{m_r} \mathbf{C}_{m_r|m_{r-1}}^{(r)}$, which are the vectorized canonical Kraus operators of CP maps $\mathcal{E}^{(r)}$ and $\sum_{m_r} \mathcal{C}_{m_r|m_{r-1}}^{(r)}$, respectively.
Hence we can express $\mathbf{E}_e$ (for $e=e_0,\dots,e_l$) and $\mathbf{I}_{m_l}$ as
\begin{equation}
\begin{gathered}
    \mathbf{E}_{e} = |E_{e}\rr \ll E_{e}| \\
    \mathbf{I}_{m_l} = \sum_{o\in O_{m_l}} |C_{m_l,o}\rr\ll C_{m_l,o}| \;.
\end{gathered}
\end{equation}
Here, $o=o_1,\dots,o_l$ is a sequence of check measurement outcomes and $O_{m_l}$ is the set of all check measurement outcome sequence $o$ resulting in final memory state $m_l$, i.e. $o\in O_{m_l}$ if and only if there exists $m_1,\dots,m_{l-1}$ such that $o=o_1,\dots,o_l$ satisfies $f_1(o_1)=m_1,\;  f_2(o_2,m_1)=m_2,\, \dots,\, f_l(o_l,m_{l-1})=m_l$.
The vectors $|E_{e}\rr$ and $|C_{m_l,o}\rr$ are defined by
\begin{equation}\label{eqn:vectorized_kraus_representation_checks_errors}
\begin{aligned}
    |C_{m_l,o}\rr &= |C_{o_l|m_{l-1}}^{(l)}\rr\otimes\dots\otimes |C_{o_1}^{(1)}\rr \\
    |E_{e}\rr &= \sum_{i_{0:l-1},j_{0:l-1}} \Big( \<i_{l-1}|_{E_{l-1}}E_{e_{l-1}}|j_{l-1},i_{l-2}\>_{Q_{l-1}E_{l-2}} \otimes\dots\otimes \<i_1|_{E_1}E_{e_1}|j_1,i_0\>_{Q_1E_0} \otimes \<i_0|_{E_0}E_{e_0}|j_0\>_{Q_0} \Big) \\ 
    &\quad \otimes E_{e_l}|j_l,i_{l-1}\>_{Q_lE_{l-1}} \otimes |j_{0:l}\>_{Q_0\dots Q_l}
\end{aligned}
\end{equation}
where $|j_{0:l}\>_{Q_0\dots Q_l} = \bigotimes_{r=0}^{l-1} |j_r\>_{Q_r}$ and for some orthonormal bases $\{|i_r\>\}_{i_r}$ and $\{|j_r\>\}_{j_r}$ of the noise environment $\mathscr{S}_{E_r}$ and codespace $\mathscr{S}_{Q_r}$, respectively.
Note that $|C_{m_l,o}\rr \in \mathscr{S}_{Q_0'}\otimes\mathscr{S}_{Q_1}\otimes\dots\otimes\mathscr{S}_{Q_{l-1}}\otimes\mathscr{S}_{Q_l}$ and $|E_{e}\rr \in \mathscr{S}_{Q_0}\otimes\mathscr{S}_{Q_0'}\otimes\dots\otimes\mathscr{S}_{Q_l}\otimes\mathscr{S}_{Q_l'}$.
Using these formulas, we can now express complete interaction between the sequence of check instruments $\mathbf{I}_{m_l}$ and error maps $\mathbf{E}_{e}$ as
\begin{equation}
\begin{gathered}
    \mathbf{E}_{e}\ast\mathbf{I}_{m_l} = \sum_{o\in O_{m_l}} |K_{e,m_l,o}\rr\ll K_{e,m_l,o}|
\end{gathered}
\end{equation}
where
\begin{equation}
\begin{aligned}
    &|K_{e,m_l,o}\rr = E_{e_l}(C_{o_l|m_{l-1}}^{(l)}\otimes I_{E_{l-1}}) E_{e_{l-1}} \dots E_{e_1} (C_{o_1}^{(1)}\otimes I_{E_0}) E_{e_0}|j_0\>_{Q_0} \otimes |j_0\>_{Q_0} \\
    &= \sum_{i_{0:l},j_{1:l},k_{1:l}} \Big( \<j_l|_{Q_l}C_{o_l|m_{l-1}}^{(l)}|k_l\>_{Q_{l-1}'} \<k_l,i_{l-1}|_{Q_{l-1}'E_{l-1}} E_{e_{l-1}}|j_{l-1},i_{l-2}\>_{Q_{l-1}E_{l-2}} \dots \<k_2,i_1|_{Q_1'E_1}E_{e_1}|j_1,i_0\>_{Q_1E_0}  \<j_1|_{Q_1}C_{o_1}^{(1)}|k_1\>_{Q_0'} 
    \\ &\quad  \<k_1,i_0|_{Q_0',E_0}E_{e_0}|j_0\>_{Q_0} \Big) E_{e_l}|j_l,i_{l-1}\>_{Q_lE_{l-1}} \otimes |j_0\>_{Q_0} \;,
\end{aligned}
\end{equation}
which is an vector in $\mathscr{S}_{Q_l'}\otimes\mathscr{S}_{Q_0}$.

We give a more explicit derivation of the operator representing the entire interaction between the errors and the check instruments and the decoding procedure.
Consider operators $\mathbf{D}_{m_l} \in \mathscr{H}_D \otimes \mathscr{H}_{Q_l'}$ and $\mathbf{E}_{e} \in \bigotimes_{r=0}^l \mathscr{H}_{Q_r'}\otimes\mathscr{H}_{Q_r}$ and $\mathbf{C}_{m_r|m_{r-1}}^{(r)} \in\mathscr{H}_{Q_r}\otimes\mathscr{H}_{Q_{r-1}'}$.
Let $Q = \{Q_0,Q_0',\dots,Q_l,Q_l',D\}$ be the set of labels of the code spaces in the dynamical code and for $\Tilde{Q}\subseteq Q$ denote $(\cdot)^{\top_{\Tilde{Q}}}$ as partial transpose over spaces with labels in $\Tilde{Q}$ and $\tr_{\Tilde{Q}}$ as partial trace over spaces with labels in $\Tilde{Q}$ and $I_{\Tilde{Q}} = \bigotimes_{Q'\in\Tilde{Q}} I_{Q'}$.
The entire dynamical encoding, error sequence, and decoding can be expressed as
\begin{equation}
\begin{aligned}
    &\sum_{m_{0:l}} \mathbf{D}_{m_l}\ast \mathbf{E}_{e_l} \ast \mathbf{C}_{m_l|m_{l-1}}^{(l)} \ast \dots \ast \mathbf{E}_{e_1} \ast \mathbf{C}_{m_1}^{(1)} \ast \mathbf{E}_{e_0} \\
    &= \sum_{m_{0:l}} \tr_{Q\backslash DQ_0}\Big( (\mathbf{D}_{m_l}\otimes I_{Q\backslash Q_l'}) \, \prod_{r=1}^l (\mathbf{E}_{e_r}^{\top_{Q_r',Q_r}}\otimes I_{Q\backslash Q_r',Q_r}) \, (\mathbf{C}_{m_r|m_{r-1}}^{(r)}\otimes I_{Q\backslash Q_r,Q_{r-1}'}) \, (\mathbf{E}_{e_0}^{\top_{Q_0',Q_0}}\otimes I_{Q\backslash Q_0',Q_0}) \Big) \\
    &= \sum_{m_{0:l}} \tr_{Q\backslash DQ_0}\Big( (\mathbf{D}_{m_l}\otimes I_{Q\backslash Q_l'}) \, (\mathbf{E}_{e}^\top \otimes I_D) \, \big( \prod_{r=1}^l \mathbf{C}_{m_r|m_{r-1}}^{(r)}\otimes I_{Q\backslash Q_r,Q_{r-1}'} \big) \Big) \\
    &= \tr_{Q\backslash DQ_0}\bigg( (\mathbf{E}_{e}^\top \otimes I_D) \, \underbrace{ \Big( \sum_{m_{0:l}} \mathbf{D}_{m_l} \bigotimes_{r=1}^l \mathbf{C}_{m_r|m_{r-1}}^{(r)} \Big) }_{\mathbf{Q}=\sum_{m_l}\mathbf{D}_{m_l}\ast\mathbf{I}_{m_l}} \bigg) \\
    &= \tr_{Q\backslash DQ_0}\Big( (\mathbf{E}_{e}^\top \otimes I_D) \, \mathbf{Q} \Big) \\
    &= \mathbf{E}_{e} \ast \mathbf{Q}
\end{aligned}
\end{equation}
Note that $\mathbf{E}_{e} \ast \mathbf{Q} \in \mathscr{H}_D \otimes \mathscr{H}_{Q_0}$.

\subsection{Strategic code with quantum memory}\label{app:interrogator_quantum_memory}

\begin{figure*}
    \centering
    \includegraphics{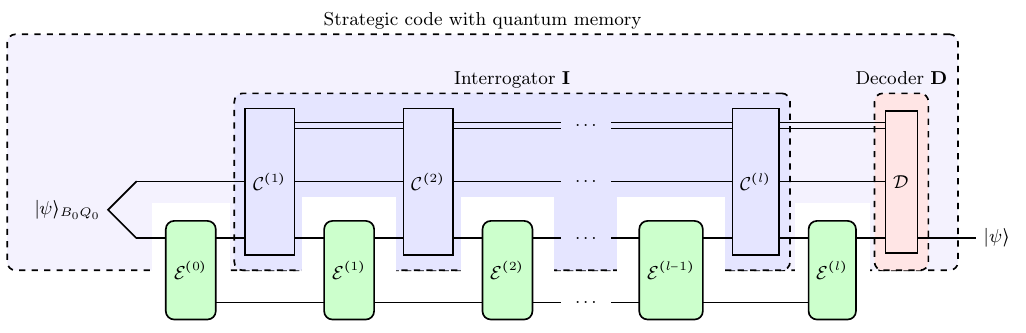}
    \caption{Strategic code with quantum memory.}
    \label{fig:interrogator_quantum_memory}
\end{figure*}

The strategic code in Definition~\ref{def:dynamical_QECC} can be generalized further to the case where retention of some quantum information can be performed between rounds.
In this case, the strategic code is equipped with a quantum memory represented by a sequence of quantum systems $\mathscr{S}_{B_0},\mathscr{S}_{B_1},\dots,\mathscr{S}_{B_l}$ where system $\mathscr{S}_{B_r}$ is the quantum system being passed from check instrument in round $r$ to the check instrument in round $r+1$ for $r\geq1$.
This is illustrated in Fig.~\ref{fig:interrogator_quantum_memory}.
At round $r=0$, without loss of generality we can think of the logical information being initially encoded in an entangled codestate $|\psi\>_{B_0Q_0}$ between the codespace $\mathscr{S}_{Q_0}$ and the quantum memory system $\mathscr{S}_{B_0}$ of the strategic code.
The $\mathscr{S}_{B_0}$ part of the entangled codestate serves as an input to check instrument $\mathscr{C}^{(1)}$ in round $1$.
In this case the check instrument in round $r$ has the form $\mathcal{C}^{(r)}:\mathscr{H}_{B_{r-1}}\otimes\mathscr{H}_{Q_{r-1}'}\rightarrow\mathscr{H}_{B_r}\otimes\mathscr{H}_{Q_r}$ since it receives the quantum system $\mathscr{S}_{B_{r-1}}$ from the preceding check instrument $\mathcal{C}^{(r-1)}$, whereas $\mathcal{C}^{(0)}$ receives the $\mathscr{S}_{B_0}$ part of the initial entangled codestate $|\psi\>_{B_0Q_0}$.

In the quantum combs representation of the interrogator $\mathbf{I}$, the eigenvectors of the interrogator operator $\mathbf{I}_{m_l}$ no longer has the tensor product structure as in the case when only classical memory is allowed in eqn.~\eqref{eqn:vectorized_kraus_representation_checks_errors}.
Namely in general we have $|C_{m_l,o}\rr \neq |C_{o_l|m_{l-1}}^{(l)}\rr\otimes\dots\otimes |C_{o_1}^{(1)}\rr$.
Here $|C_{m_l,o}\rr$ instead takes the more general form of
\begin{equation}\label{eqn:quantum_memory_interrogator_vectorized_kraus}
\begin{aligned}
    |C_{m_l,o}\rr &= \sum_{k_{0:l-1},j_{0:l-1}} C_{o_l,m_{l-1}}|k_{l-1},j_{l-1}\>_{B_{l-1}Q_{l-1}'} \otimes\dots\otimes \<k_2|_{B_2}C_{o_2|m_1}|k_1,j_1\>_{B_1Q_1'} \otimes \<k_1|_{B_1}C_{o_1}|k_0,j_0\>_{B_0Q_0'} \\ 
    &\quad \otimes |k_0\>_{B_0} \otimes |j_{0:l-1}\>_{Q_0'\dots Q_l'}
\end{aligned}
\end{equation}
where $|j_{0:l-1}\>_{Q_0'\dots Q_l'} = \bigotimes_{r=0}^{l-1} |j_r\>_{Q_r'}$ and for some orthonormal bases $\{|k_r\>\}_{k_r}$ and $\{|j_r\>\}_{j_r}$ of the quantum memory system $\mathscr{S}_{B_r}$ and codespace $\mathscr{S}_{Q_r'}$, respectively.
Also as before $m_r = f_r(o_r,m_{r-1})$ for all $r>1$ and $m_1=f_1(o_1)$.

Note that when we set the number of rounds of the strategic code to $l=0$, we recover the entanglement-assisted QECC (EAQECC)~\cite{brun2006correcting,hsieh2007general,brun2014catalytic}.
In this case we simply have the initial entangled codestate $|\psi\>_{B_0Q_0}$ followed by error map $\mathcal{E}^{(0)}$ then a decoder channel $\mathcal{D}$, which makes up an EAQECC.

It is an interesting future work to establish a necessary and sufficient error-correction conditions for a strategic code with quantum memory analogous to Theorem~\ref{thm:algebraic_KL_condition} and Theorem~\ref{thm:info_theoretic_condition}. 
In doing this one needs to restrict the dimension of the quantum memory of the interrogator, as otherwise one can always store the entire code in the quantum memory, bypassing the error maps.
Also, one might also consider where the decoder also outputs a ``residue'' entangled state alongside the recovered initial codestate as in~\cite{grassl2022entropic}.

\subsection{Subsystem strategic code}\label{app:strategic_subsystem_code}

Another generalization of the strategic code in Definition~\ref{def:dynamical_QECC} is to introduce additional system in the codespace in each round, i.e. $\mathscr{S}_{Q_r} = \mathscr{S}_{A_r}\otimes\mathscr{S}_{C_r}$ where logical information is stored in subsystem $\mathscr{S}_{A_r}$.
This is analogous to subsystem QECC~\cite{kribs2005unified,kribs2005operator,poulin2005stabilizer,nielsen2007algebraic} where the codespace is of the form $\mathscr{S}_{Q} = \mathscr{S}_{A}\otimes\mathscr{S}_{C}$.
In this generalization, which we call a \textit{subsystem strategic code}, we still retain the form of the interrogator operator $\mathbf{I}_{m_l} = \sum_{o\in O_{m_l}} |C_{m_l,o}\rr\ll C_{m_l,o}|$.
Namely, $|C_{m_l,o}\rr$ has the same expression as eqn.~\eqref{eqn:vectorized_kraus_representation_checks_errors}, or as eqn.~\eqref{eqn:quantum_memory_interrogator_vectorized_kraus} in the case where quantum memory is available
Hence we can modify Definition~\ref{def:dynamical_QECC} so that we say a subsystem strategic code $(\mathscr{S}_{Q_0},\mathbf{I})$ corrects $\mathfrak{E}$ if
\begin{equation}
    \mathcal{D}_{m_l}(\mathbf{E} \ast \mathbf{I}_{m_l} \ast (\rho\otimes\sigma)) = \rho \otimes \sigma_{m_l}
\end{equation}
for all density operators $\rho$ in $\mathscr{H}_{A_0}$ and $\sigma,\sigma_{m_l}$ operators in $\mathscr{H}_{C_0})$.
Lastly, necessary and sufficient conditions in Theorem~\ref{thm:algebraic_KL_condition} generalized to the subsystem strategic codes should also reduce to the necessary and sufficient condition for subsystem codes~\cite{nielsen2007algebraic} when we set the number of rounds $l=0$,
\begin{equation}
    \Pi_QE_{e'}^\dag E_e\Pi_Q = \Pi_A \otimes g_{e',e}
\end{equation}
where $g_{e',e}$ is an operator in $\mathscr{H}_C$.
This generalized condition for subsystem strategic code, however, is left for future work.

\subsection{Virtual strategic code and strategic codes with exotic causal structure}\label{app:virtual_strategic_code}

A recently proposed virtual quantum resource theory~\cite{yuan2024virtual,takagi2024virtual} offers a framework of approximating a process $\Phi$ by performing sampling process from a set of allowed process $\mathscr{F}$ where $\Phi\notin\mathscr{F}$, followed by a post-processing, to achieve a certain task.
As it is shown in~\cite{takagi2024virtual} on how this framework can be applied to quantum combs, here we show how one can ``virtualize'' a strategic code, including those with more exotic causal structure such as an indefinite causal order or causal inseparability~\cite{oreshkov2012quantum,chiribella2013quantum,procopio2015experimental,oreshkov2016causal,costa2016quantum,milz2018entanglement,rubino2017experimental,goswami2018indefinite,loizeau2020channel,ebler2018enhanced}.

First we can consider a set of $l$-rounds allowed strategic codes $\{\mathbf{I}^{(k)}\}_k$ (e.g. those that only maintains classical memory), where we include an encoding channel $\mathcal{G}^{(k)}:\mathscr{L}(\C^{d'})\rightarrow\mathscr{H}_{Q_0}^{(k)}$ mapping linear operators on $d'$ dimensional complex vector space to linear operators on an initial codespace $\mathscr{S}_{Q_0}^{(k)}$ in the strategic code $\mathbf{I}^{(k)}$.
So a $d'$ dimensional state $|\psi\>$ encoded with strategic code $\mathbf{I}^{(k)}$ with error $\mathbf{E}$ gives a state
\begin{equation}
    \mathbf{E} \ast \mathbf{I}_{m_l}^{(k)} \ast \ketbra{\psi}
\end{equation}
at the start of the decoding round for a final memory state $m_l$.

The sampling and post-processing processes are based on linear expansion of operator $\mathbf{\Phi} = \sum_k \beta_k \mathbf{I}^{(k)}$ where $\beta_k\in\R$ where $\mathbf{\Phi}$ represents some black-box process that allows input of initial state $\ketbra{\psi}$ and interaction with error $\mathbf{E}$ as $\mathbf{I}^{(k)}$ do.
Operator $\mathbf{\Phi}$ have the same dimension as $\mathbf{I}^{(k)}$, but it may not correspond to a quantum comb.
Namely it may correspond to a process involving indefinite causal order or causally inseparable, e.g. where the effect of errors between round $r$ and $r'\neq r$ may not have a definite causal relation (as opposed to the strategic code where interaction between the investigator's operation and the error map at round $r=1$ influences the interaction at round $r=3$, but not the other way around).
Here, operator $\mathbf{\Phi}$ is represented by a process matrix~\cite{oreshkov2012quantum,oreshkov2016causal,costa2016quantum}, which is more general object than a quantum comb.

To perform sampling, one constructs a probability distribution over $k$ with probabilities $\gamma_k = \frac{|\beta_k|}{\tau}$ for $\tau = \sum_k |\beta_k|$ so that
\begin{equation}
    \mathbf{\Phi} = \sum_k (\mathrm{sign}(\beta_k)\tau) \gamma_k \mathbf{I}^{(k)} \;.
\end{equation}
Using this relation, one can sample from distribution $\{\gamma_k\}_k$ and upon obtaining outcome $k$, use strategy code $\mathbf{I}^{(k)}$ to encode some fixed state $|\psi\>$ and after applying noise $\mathbf{E}$ obtain the output state $\mathbf{E} \ast \mathbf{I}_{m_l}^{(k)} \ast \ketbra{\psi}$ then apply ``post-processing'' by a multiplication by $(\mathrm{sign}(\beta_k)\tau)$.
As shown in~\cite{yuan2024virtual,takagi2024virtual}, by performing this sampling and post-processing multiple times we can obtain an approximation of 
\begin{equation}
    \tr\Big( (\mathbf{E} \ast \mathbf{\Phi} \ast \ketbra{\psi}) A \Big)
\end{equation}
for some bounded linear operator $A$.
Lastly, we note that decoder may be included in $\mathbf{I}^{(k)}$ as well to have the entire QECC process where the output $\mathbf{E} \ast \mathbf{I}_{m_l}^{(k)} \ast \ketbra{\psi}$ is the output of a decoder.
It would be interesting to investigate into the performance of strategic code virtualization and how strategic codes with exotic causal structures can or cannot improve code performance.
However, this is left for future work.

\section{Spacetime Code in the Quantum Combs Formalism}\label{app:spacetime_code}

Now we describe the spacetime code~\cite{bacon2017sparse,gottesman2022opportunities,delfosse2023spacetime} in our dynamical QECC quantum combs framework.
The spacetime code is first proposed by Bacon,et.al. for a circuit consisting of a subset of Clifford gates in~\cite{bacon2017sparse}, then a generalization to circuits consisting of any Clifford gates is done by Gottesman in~\cite{gottesman2022opportunities}.
In these two spacetime codes, Pauli measurements for syndrome extraction is performed after the last layer of Clifford gates is applied.
This is later generalized by Delfosse,et.al. in~\cite{delfosse2023spacetime} where Pauli measurements can be performed anywhere in the circuit.
Here we consider the most general spacetime code defined in~\cite{delfosse2023spacetime} in demonstrating how spacetime code fits in our framework.

A spacetime code is defined by a circuit that takes $n$ qubits as input, followed by $l$ layers of Clifford operations, where each layer consists of Clifford gates and Pauli measurements on disjoint qubits.
In Gottesman's and Bacon,et.al.'s spacetime code~\cite{gottesman2022opportunities} where measurements are restricted to the end of the circuit, the circuit takes $q$ qubit input along with preparation of $a$ ancilla qubits at the beginning of the circuit and $q'$ output qubits with $\{|0\>,|1\>\}$ basis measurements on $b$ qubits at the end of the circuit.
In this circuit $q,a,q',b$ must satisfy $n=q+a=q'+b$, where $n$ is the width of the circuit and thus each layer consists only of Clifford gates.

We note that our dynamical QECC framework can also describe a \textit{sequence} of such circuits $\mathbf{C}^{(1)},\mathbf{C}^{(2)},\dots,\mathbf{C}^{(l)}$ where error syndromes from one circuit determines the structure of the subsequent circuit, hence induces adaptivity.
This temporal dependence across circuits has been mentioned in~\cite{gottesman2022opportunities} although was not explored further.
The quantum combs formalism for dynamical QECC applied to spacetime code presented here allows such exploration with the natural temporal-dependence representation.

\begin{figure*}
    \centering
    \includegraphics{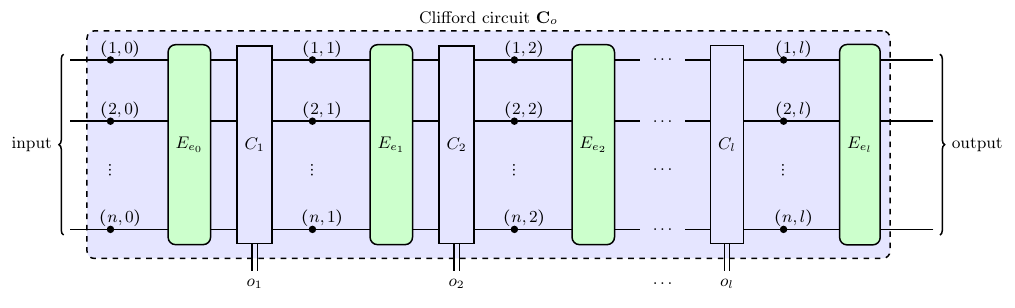}
    \caption{
    Spacetime code Clifford circuit.
    Circuit $\mathbf{C}_o$ takes $n$ qubits as input at the start of the circuit and outputs $n$ qubits along with measurement outcomes $o=o_1,\dots,o_l$ peformed throughout the circuit.
    Between the input and the output, the circuit contains $l$ layers of Clifford operations $C_1,\dots,C_l$, each layer $C_r$ consist of Clifford gates and Pauli measurements with acting on disjoint subsets of the $n$ qubits.
    Outcomes from Pauli measurements at layer $r$ is denoted by $o_r$.
    Error operations are modelled to occur after the input ($E_{e_0}$) and after each layer ($E_{e_1},\dots,E_{e_l}$).
    Error $E_{e_0}$ represents the noise on input qubits and error $E_{e_r}$ for $r\geq1$ represent noise from the Clifford gates and Pauli measurements in layer $C_r$.
    A spacetime code for circuit $\mathbf{C}$ is defined by $n(l+1)$-qubit stabilizer group, which in turn is defined by the circuit components.
    Each qubit (illustrated as black dots) is labeled by $(i,r)$: for $r\geq1$ it correspond to the $i$-th qubit output of $C_r$ and for $r=0$ it correspond to the $i$-th qubit of the input.
    Error $E_{e_r}$ in round $r$ is applied to qubits $\{(i,r)\}_i$.
    }
    \label{fig:spacetime_circuit}
\end{figure*}

\subsection{Dynamical QECC quantum combs representation of spacetime code}

Now we describe the quantum combs dynamical QECC of a spacetime code with respect to a circuit with $l$ layers (see Fig.~\ref{fig:spacetime_circuit}).
In our dynamical QECC framework, this spacetime code has $l$ rounds. 
At round $r=0$, error $E_{e_0}$ is inflicted at the $n$ qubit input.
At round $r\geq1$, Clifford operation $C_r$ is applied to the $n$ qubits, followed by error $E_{e_r}$.
The Clifford operation consisting of Clifford gates and Pauli measurements performed on disjoint set of qubits is described by $C_r = \{C_{r,o_r}\}_{o_r\in O_r}$ where $C_{r,o_r}$ is a bounded linear operator from $\C^{2^n}$ to $\C^{2^n}$ and $o_r\in O_r$ is a measurement outcome from the Pauli measurements with $O_r$ being the set of all possible measurement outcomes.
If there are $k$ Pauli measurements in $C_r$ then the outcome is in the form of a $k$-tuple $o_r=(o_{r_1},\dots,o_{r_k})$. 
If no Pauli measurement is performed at round $r$, then we set a constant outcome for this round, i.e. $O_r=\{o_r\}$ is a singleton set and $C_{r,o_r}$ is a Clifford unitary.
Clifford operation $C_r$ at each round then defines a quantum instrument $\mathcal{C}_r = \{\mathcal{C}_{r,o_r}\}_{o_r}$ where $\mathcal{C}_{r,o_r}(\rho) = C_{r,o_r} \rho C_{r,o_r}^\dag$ is a CP map such that $\sum_{o_r} \mathcal{C}_{r,o_r}$ is a CPTP map.
When measurement outcome $o=o_1,\dots,o_l$ is obtained, the quantum combs representation of the circuit is
\begin{equation}\label{eqn:spacetime_circuit_combs}
\begin{aligned}
    \mathbf{C}_o = |C_o\rr\ll C_o| = \bigotimes_{r=1}^l |C_{r,o_r}\rr\ll C_{r,o_r}|
\end{aligned}
\end{equation}
where $|C_{r,o_r}\rr$ is the vectorized form of operator $C_{r,o_r}$.

Errors occurring throughout the circuit is described by a sequence of bounded linear operators $E_{e_0},\dots,E_{e_l}$ mapping vectors in $\C^{2^n}$ to $\C^{2^n}$.
An error operator $E_{e_r}$ takes the form of a tensor product of Paulis on qubits labeled by $\{(i,r)\}_{i\in[n]}$.
Explicitly this can be expressed as $E_{e_r}=E_{1,e_r} \otimes \dots\otimes E_{n,e_r}$ where $E_{i,e_r}\in \{I,X,Y,Z\}$ is a qubit Pauli operator where identity $E_{i,e_r}=I$ indicates no error is inflicted on qubit $i$ at round $r$.
In the quantum combs representation, we can express the error sequence $e=e_0,\dots,e_l$ as a positive semidefinite operator $|E_e\rr\ll E_e| = \bigotimes_{r=0}^l |E_{e_r}\rr\ll E_{e_r}|$ where $|E_{e_r}\rr = \bigotimes_{i=1}^n |E_{i,e_r}\rr$ and $|E_{i,e_r}\rr = \sum_{j',j} \<j'|E_{i,e_r}|j\> |j'\>_{r,i} |j\>_{r,i}$.
Hence the combs representation of the an error sequence can be expressed conveniently as a tensor product of Choi operator of error operators for each coordinate $(i,r)$ in the spacetime grid
\begin{equation}
\begin{aligned}
    |E_e\rr\ll E_e| = \bigotimes_{r=0}^l \bigotimes_{i=1}^n |E_{i,e_r}\rr\ll E_{i.e_r}| \;,
\end{aligned}
\end{equation}
which closely resembles how errors are represented in spacetime code as a tensor product of the error operators $F_e = \bigotimes_{r=0}^l \bigotimes_{i=1}^n E_{i,e_r}$.
Using both the quantum combs representation of the circuit (eqn.\eqref{eqn:spacetime_circuit_combs}) and the error sequence, the interaction between then can be expressed using the link product as
\begin{equation}
\begin{aligned}
    |E_e\rr\ll E_e| \ast \mathbf{C}_o
\end{aligned}
\end{equation}
which is a positive semidefinite operator from $\C^{2^{2n}}$ to $\C^{2^{2n}}$ corresponding to a CP map from the $n$-qubit input to the $n$-qubit output.

Spacetime code is then defined by $n(l+1)$ qubit stabilizer group $\mathscr{S}_\mathrm{st}$ where the qubits are placed in a grid and labeled by a tuple $(i,r)$ where $i\in\{1,\dots,n\}$ corresponds to a qubit register and $r\in \{0,\dots,l\}$ corresponds to a layer in the circuit.
So, qubits labeled by $\{(i,0)\}_i$ are input qubits which error $E_{e_0}$ are inflicted upon.
Whereas for $r\geq1$, qubits $\{(i,r)\}_i$ are the qubits placed at the output of Clifford operation $C_r$ which error $E_{e_r}$ occurs.
If there are $s$ Pauli measurements across the circuit with observables $S_1,\dots,S_s$, then we can write the collection of the outcomes as a bit string $o=o_1,\dots,o_s$ which is then being put through a function $f(o)=m$ which maps measurement outcomes to error syndromes.
Error syndrome $m$ is then used to choose which of the decoding channel $\{\mathcal{D}_{m}\}_m$ should be used at the circuit output.

In Delfosse,et.al.'s spacetime code~\cite{delfosse2023spacetime} decoding is done in three steps.
The first step is by using check bitstrings $\{u_1,\dots,u_q\}\subseteq\{0,1\}^s$ to define the function $f$ mapping the measurement outcome bitstring $o$ to syndrome bitstring $m=m_1\dots m_q\in\{0,1\}^q$.
This is done by taking the inner product between $u_j$ and $o$ as the $j$-th bit of $f(o)=m$, i.e. $m_j = \langle u_j, o\rangle = u_{j,1}o_1 + \dots + u_{j,s}o_s$, which is to be understood as the $j$-th syndrome of measurement outcome $o$.
The second step is to identify the circuit error $F_e = \bigotimes_{r=0}^l \bigotimes_{i=1}^n E_{i,e_r}$ (or $|E_e\rr\ll E_e|$ in the quantum combs form) using syndrome $m$.
In~\cite{delfosse2023spacetime}, this is done by using a ``most likely effect'' (MLE) decoder on the $nl$-qubit spacetime stabilizer code.
This stabilizer code is defined by circuit $\{\mathbf{C}_o\}_o$ and check bitstrings $\{u_1,\dots,u_q\}$ (or equivalently, function $f$).
Given syndrome $m$, the MLE decoder outputs an $nl$-qubit Pauli $g(m)$ which is its guess of the true circuit error $F_e$.
Now the third step is to propagate $g(m)$, denoted by $\overrightarrow{g(m)}$, and take its $n$-qubit Pauli corresponding to the last layer $l$ of the circuit, denoted by $[\overrightarrow{g(m)}]_l$.
Propagation of (a subset of) an $n(l+1)$-qubit Pauli $P$ placed on the grid of the circuit is introduced as the spackle operation in~\cite{bacon2017sparse}.
Then we compute commutation relation between $\overrightarrow{g(m)}$ and observables $S_1,\dots,S_s$ corresponding to each measurement, to obtain bitstring $\gamma=\gamma_1,\dots,\gamma_s$ where $\gamma_j = [S_j,\overrightarrow{g(m)}]$.
Then we can unflip measurement outcome string $o$ by $o+\gamma$.
Using $[\overrightarrow{g(m)}]_l$ and $\gamma$ we can correct the output state from the circuit.

\subsection{Spacetime code generalization in the strategic code framework}

In the strategic code framework a generalization of the spacetime code can be constructed using adaptivity on the Clifford operations.
Following the notation of the strategic code, we can denote the set of allowed Clifford operations in round $r$ as $\{\mathcal{C}_{m_{r-1}}^{(r)}\}_{m_{r-1}}$, where $\mathcal{C}_{m_{r-1}}^{(r)}$ is a CP map defined by its action $\mathcal{C}_{m_{r-1}}^{(r)}(\rho) = \sum_{o_r} C_{o_r|m_{r-1}}^{(r)} \rho C_{o_r|m_{r-1}}^{(r)\dag}$.
As before, $C_{o_r|m_{r-1}}^{(r)}$ consists of Clifford gates and Pauli measurements on disjoint set of qubits.
For simplicity, we assume that $m_r$ has a one-to-one correspondence with the sequence of measurement outcomes $o_1,\dots,o_r$ hence we can still describe the circuit as $\mathbf{C}=\{\mathbf{C}_o\}_o$ as a sequence of outcomes $o=o_1,\dots,o_l$ is maintained in the classical memory of the code until the decoding round.
For example, when a Pauli measurement $S_1$ in round $1$ gives an outcome $o_1=+1$, we apply Clifford operation $\mathcal{C}_{+1}^{(2)}(\cdot) = \sum_{o_2} C_{o_2|+1}^{(2)} \cdot C_{o_2|+1}^{(2)\dag}$ in round $2$, otherwise we apply Clifford operation $\mathcal{C}_{-1}^{(2)}(\cdot) = \sum_{o_2} C_{o_2|-1}^{(2)} \cdot C_{o_2|-1}^{(2)\dag}$ where Clifford gates and Pauli measurements in $C_{o_2|+1}^{(2)}$ and $C_{o_2|-1}^{(2)}$ may differ.

In this case we have a family of spacetime codes, one for each trajectory defined by a sequence of measurement outcomes $o$.
Using the same method in obtaining the stabilizer group of the spacetime code (i.e. from the backpropagated measurement observables), we obtain a stabilizer group for each outcome sequence $o$.
A simple case where this may happen can be seen using the example mentioned at the end of the previous paragraph.
Suppose $C_{+1|+1}^{(2)} = U_2\otimes \ketbra{0}$ and $C_{+1|-1}^{(2)} = U_2 \otimes \ketbra{+}$, then their corresponding measurement observables are $Z$ and $X$.

\section{Hastings-Haah honeycomb Floquet code as a strategic code}\label{app:hastings_haah_interrogator}

\begin{figure*}
    \centering
    \includegraphics{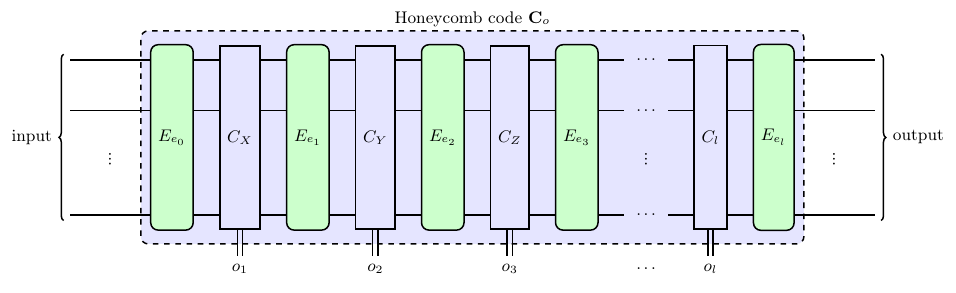}
    \caption{
    Hastings-Haah honeycomb code interrogator.
    }
    \label{fig:hastings_haah_interrogator}
\end{figure*}

Here we look into how the strategic code framework can be used in determining correctable error in the Hastings-Haah honeycomb Floquet code.
For simplicity, we consider a single hexagon in the honeycomb code with corresponding stabilizer $ZZZZZZ$.
Measurement observables in round 1 and 2 are
\begin{equation}
\begin{aligned}
    XXIIII, IIXXII, IIIIXX \\
    IYYIII, IIIYYI, YIIIIY \;,
\end{aligned}
\end{equation}
respectively, and the outcomes are $o_1,o_2\in\{+1,-1\}^3$, where $\{+1,-1\}$ are the corresponding eigenvalues of the two eigenstates for each observable.
Hence the check instruments $\mathcal{C}_X^{(1)}$ and $\mathcal{C}_Y^{(2)}$ at these two rounds have Kraus operators
\begin{equation}
\begin{gathered}
    C_{+1,+1,+1|X}^{(1)} = \big(\ketbra{++}+\ketbra{--}\big)_{1,2} \otimes \big(\ketbra{++}+\ketbra{--}\big)_{3,4} \otimes \big(\ketbra{++}+\ketbra{--}\big)_{5,6} \\
    \vdots \\
    C_{-1,-1,-1|X}^{(1)} = \big(\ketbra{+-}+\ketbra{-+}\big)_{1,2} \otimes \big(\ketbra{+-}+\ketbra{-+}\big)_{3,4} \otimes \big(\ketbra{+-}+\ketbra{-+}\big)_{5,6} \\
    C_{+1,+1,+1|Y}^{(2)} = \big(\ketbra{+i,+i}+\ketbra{-i,-i}\big)_{2,3} \otimes \big(\ketbra{+i,+i}+\ketbra{-i,-i}\big)_{4,5} \otimes \big(\ketbra{+i,+i}+\ketbra{-i,-i}\big)_{6,1} \\
    \vdots \\
    C_{-1,-1,-1|Y}^{(2)} = \big(\ketbra{+i,-i}+\ketbra{-i,+i}\big)_{2,3} \otimes \big(\ketbra{+i,-i}+\ketbra{-i,+i}\big)_{4,5} \otimes \big(\ketbra{+i,-i}+\ketbra{-i,+i}\big)_{6,1}
\end{gathered}
\end{equation}
where subscripts indicates which qubits the projector is acting on and $|\pm\pm\> = |\pm\>\otimes|\pm\>$ and $|\pm i,\pm i\> = |\pm i\> \otimes |\pm i\>$ and $|\pm\>,|\pm i\>$ are $\pm 1$ eigenstates of Pauli $X$ and $Y$, respectively.

Consider two errors $E_{e_0}=ZIIIII$ and $E_{e_0'}=IZIIII$ at round $0$ (recall that round $r$ error occurs after round $r$ operation and before round $r+1$ operation), and no further errors occur in round 1 and 2 i.e. $E_{e_1}=E_{e_2}=IIIIII$.
Both of these errors are correctable by the honeycomb code as both errors $E_{e_0}=ZIIIII$ and $E_{e_0'}=IZIIII$ flips the outcome of both round 1 and round 2 measurements, allowing the decoder to detect the error.
In the strategic code framework, we have a pair of error sequences $e=e_0,e_1,e_2$ and $e'=e_0',e_1,e_2$ with corresponding vectorized error operators
\begin{equation}
\begin{gathered}
    |E_e\rr = |ZIIIII\rr \otimes |IIIIII\rr \otimes |IIIIII\rr \\
    |E_e'\rr = |IZIIII\rr \otimes |IIIIII\rr \otimes |IIIIII\rr \;.
\end{gathered}
\end{equation}
Consider orthogonal states $|j\>,|k\>$ 
in the initial codespace $\mathscr{S}_{Q_0}$.
Then since all check outcomes are stored in the memory of the honeycomb code interrogator, by Corollary~\ref{cor:algebraic_condition_all_outcome_memory} it holds that
\begin{equation}
\begin{aligned}
    \ll E_{e'}|(|C_o\rr\ll C_o| \otimes |k\>\<j|)|E_e\rr = \delta_{k,j} \lambda_{e,e',o}
\end{aligned}
\end{equation}
for some constant $\lambda_{e,e',o}$.
Since the $Z$ pauli flips the $|\pm\>$ to $|\mp\>$ and $|\pm i\>$ to $|\mp i\>$, we obtain outcomes $o_1 = (-1,+1,+1)$ and $o_2=(+1,+1,-1)$ for $E_e$ and $o_1 = (-1,+1,+1)$ and $o_2=(-1,+1,+1)$ for $E_{e'}$ (since without error both initial codestate $|j\>$ and $|k\>$ gives all $+$ outcomes for both $o_1$ and $o_2$)

We can verify that this condition holds for error sequences $e,e'$ as
\begin{equation}
\begin{aligned}
    \ll E_{e'}|(|C_o\rr\ll C_o| \otimes |k\>\<j|)|E_e\rr &= \<k|E_{e_0'}C_{o_1}^{(1)\dag}C_{o_2}^{(2)\dag}C_{o_2}^{(2)}C_{o_1}^{(1)}E_{e_0}|j\> = 0 
\end{aligned}
\end{equation}
for all $o=o_1,o_2$.
One can see this by noting that $C_{o_2}^{(2)\dag}C_{o_2}^{(2)}C_{o_1}^{(1)}E_{e_0}|j\>=0$ for all $o_1 \neq (-1,+1,+1)$ and $o_2\neq(+1,+1,-1)$ while $\<k|E_{e_0'}C_{o_1}^{(1)\dag}C_{o_2}^{(2)\dag}=0$ for all $o_1 \neq (-1,+1,+1)$ and $o_2\neq(-1,+1,+1)$.
Namely, the sequence of check measurement outcomes maps $E_{e_0}|j\>$ and $E_{e_0'}|k\>$ to different subspaces (even with $j=k$), which allows one to construct a decoder detecting and distinguishing these errors.

\end{document}